\documentclass[%
 reprint,
 amsmath,amssymb,
 aps,
]{revtex4-2}

\usepackage{amsthm}
\usepackage{mathtools}
\usepackage{color}
\usepackage{amsmath}
\usepackage{mathrsfs}
\usepackage{graphicx}
\usepackage{dcolumn}
\usepackage{bm}

 \newcommand{\HC}{\mathcal{H}}

 \newcommand{\Tr}{{\rm Tr}}

\renewcommand{\geq}{\geqslant}
\renewcommand{\leq}{\leqslant}

\newcommand{\mOmega}{\mathbf \Omega}
\newcommand{\mOt}{\mathscr O^t}
\newcommand{\mO}{\mathscr O}

\newcommand{\bs}{\textsf{BS}}

\newcommand{\mI}{\mathbb I}
\newcommand{\mX}{\mathbb X}

\def\be{\begin{equation}}
\def\ee{\end{equation}}
\def\bs{\begin{split}}
\def\e{\end{split}}
\def\ba{\begin{eqnarray}}
\def\bea{\begin{eqnarray}}

\def\tea{\end{eqnarray}}
\def\ea{\end{eqnarray}}
\def\eea{\end{eqnarray}}

\newtheorem{corollary}{Corollary}

\newtheorem{proposition}{Proposition}

\newtheorem{definition}{Definition}

\begin{document}

\preprint{LA-URed}

\title{Cycle equivalence classes, orthogonal Weingarten calculus,\\ and the mean field theory of memristive systems}
\thanks{email: caravelli@lanl.gov.}%

\author{F. Caravelli}
\affiliation{Theoretical Division (T4), Los Alamos National Laboratory, Los Alamos, New Mexico 87545, USA
}%

\date{\today}

\begin{abstract}
  It has been recently noted that for a class of dynamical systems with explicit conservation laws represented via projector operators, the dynamics can be understood in terms of lower dimensional equations. This is the case, for instance, of memristive circuits. Memristive systems are important classes of devices with wide-ranging applications in electronic circuits, artificial neural networks, and memory storage. We show that such mean-field theories can emerge from averages over the group of orthogonal matrices, interpreted as  cycle-preserving transformations applied to the projector operator describing Kirchhoff's laws.  Our results provide insights into the fundamental principles underlying the behavior of resistive and memristive circuits and highlight the importance of conservation laws for their mean-field theories. In addition, we argue that our results shed light on the nature of the critical avalanches observed in quasi-two dimensional nanowires as boundary phenomena.

\end{abstract}

\maketitle


\section{Introduction}

 With the increasing interest in developing neuromorphic computers, it is crucial to understand physical devices that exhibit similar structures and functionalities to biological neural networks.  The nonlinear interactions and complex connectivity of biological neuronal networks are well-known characteristics \cite{ndbook}.  Yet, it is still a mystery why low-dimensional representations of certain rather complex dynamical systems exist \cite{lowdim1,lowdim2} even if in certain regimes. This study sheds light on the existence of such representations in the context of memristive circuits, which we use as a toy model for more complex neuromorphic systems \cite{mead}. 
 
 Circuits composed of nanodevices with memory are at the forefront of neuromorphic computing research, as their behavior often mimics synaptic plasticity observed in biological neuronal circuits. However, a comprehensive theory that effectively describes the behavior of these circuits is currently lacking. Memristive devices are a promising area of research for the development of next-generation computing systems \cite{stru,reviewCarCar}. These devices are passive 1-port resistive components that have the ability to remember past voltages and currents and can change their resistance based on the history of the input signals \cite{chua71,chua76a}. The experimental existence of switching in physical systems dates back to the late 60's \cite{argall}, but the connection to memristive behavior has been made just over a decade ago \cite{stru08,memr1,memr2}. One peculiar feature of the memristive devices is that their resistance changes between two limiting values $R_{on}\leq R_{off}$ (or analogously conductance value).
 The development of circuits of memristive devices has become an important area of research, as it enables the creation of neuromorphic devices that can support the existing von Neumann architecture \cite{DiVentraPershin,reviewCarCar,stru} in a variety of tasks more prone to an analog computing approach. From an experimental perspective,
nanowire networks have emerged as a promising material for the fabrication of disordered memristive networks. They exhibit reversible resistance switching behavior when subjected to an external electric field, making them ideal for use in memristive devices. Additionally, silver nanowires are low-cost, have high aspect ratios, and can be synthesized using a variety of techniques, making them highly versatile. The avoided crossings between the nanowires act as tunneling junctions \cite{nakayama,Zdenkaadvphys}, and for coated silver nanowires the phenomenon of filament formation is the main driver between the memristive effects that emerge \cite{AtomicSwitch1,AtomicSwitch2,nakayama,milano2020}. 
 Their behavior is particularly similar to the behavior of neuronal circuits, first and foremost, their connectivity strongly resembles the one of a neuronal connectome. In addition, the formation of filaments and their effective memristive behavior strongly parallels the plasticity of neuronal circuits. In this respect, then, memristive networks are an area of research in Physics that parallels the study of neuronal networks in biology. For instance, memristive networks have Lyapunov functions \cite{CaravelliEntropy} similar to recurrent neural networks \cite{hopfieldtank}. Studying memristive circuits can provide important insights into the non-trivial dynamics of large assemblies of neuronal networks.

As an example of this, it became clear only recently that certain disordered circuits with memory exhibit a certain regularity \cite{chua71,stru08,Caravelli2016rl,Caravelli2016ml,Caravelli2019} in their dynamical behavior as the applied voltage is constant in time. 
This should be somewhat surprising, given that these are rather nonlinearly interacting systems, characterized by conservation laws that lead to effective nonlocal interactions, although exponentially bounded \cite{Caravelli2017}. Such mean field techniques have also been tested in experimentally viable systems such as nanowire connectomes \cite{caravelliadp}. Although it is very well-known that circuits composed of purely memristive devices must also be memristive  in 2-probe experiments \cite{chua76a}, it is less obvious why the collective behavior should be so much similar to a single device, in particular when the network of devices is disordered. This observation is not only numerical or model-based, but it has been recently shown to be true also in the experimental setup of silver nanowires \cite{caravelliadp}, where a mean field theoretical description of a disordered network of nanowires could fit the potentiation and depression of the conductance.
Following these results, the author proposed that there is a class of dynamical systems that, thanks to the presence of projector operators in their dynamics, have a well-defined mean field theory \cite{caravelliPEDS}. However, rigorous results could be obtained only for simple projector operators.

On a side note, understanding the phenomenology of memristive networks has direct applications for the employment of these physical systems as reservoir computing devices at the edge of chaos, an idea proposed in the early 90's
\cite{packard,langton}, which has seen a direct application in memristive circuits \cite{res1,res2,carroll,sheldonmemlrc}. In fact, memristive systems seem to perform better near critical transitions \cite{memorynanw,baccetti}. In the case of the brain, the idea that the brain is in a constant critical state has been proposed in the mid 90's \cite{bak1,bak2,bak3,chialvo4,chialvomem,jensen_1998}.

However, in the case of memristive circuits, one can have both first and second-order dynamical transitions in the conductance state. Avalanches of memristive switching are a phenomenology also discussed in the literature \cite{avizienis,hochstetteretal,brownavalanche} both theoretically and experimentally. However, it is worth pointing out that these switching phenomena are strongly due to what we wish to call here \textit{boundary} phenomena: they depend on how the memristive device approaches dynamically the boundary of their resistive or conductance values. We call instead a \textit{bulk} phenomenon the switching is purely due to the presence of hysteresis, e.g. a pinched hysteresis loop. Boundary phenomena are based on a rapid switching of memristive devices between two resistive states. Bulk phenomena are instead due only to the presence of memory.
This is for instance the case of the recent symmetry-breaking transition of \cite{caravelliscience}. There,  a first order transition in the conductance state is purely due to the presence of hysteresis, rather than the phenomenology of the boundary of the device. In fact, the presence or not of this transition depends only on the ratio $r=R_{off}/R_{on}$, which at present is the best indicator of a bulk-induced transition. Previously, a mean-field theory for this type of transition was developed in a series of works \cite{caravelliscience,caravelliadp,caravelliPEDS}. Other mean field theories have appeared in the literature \cite{Caravelli_2018,CaravelliEntropy,pershinisc}, describing the Ising-like behavior.
The present manuscript is an attempt of overcoming some of the limitations of previous works, incorporating the circuit properties into mean field theory. Bulk 
in deriving a mean-field theory for memristive network bulk transitions, by averaging only in a subclass with the same number of cycles. This is important for experimental reasons. We know for instance that in quasi-two dimensional materials that these second order transitions are only a cross over \cite{caravelliadp}. One reason is that the mean field theoretical results imply implicitly the homogeneity of the network. We thus wish to go beyond this approximation, applying techniques that keep track of the density of cycles in the network. This result is important, even if for a toy model, to develop techniques to explain the phenomenology of more realistic systems.

The paper is organized as follows. First, we establish a connection between cycle-preserving transformations and the group of orthogonal transformations, introducing an equivalence relation on the cycle space. We then utilize the Weingarten calculus to average over orthogonal transformations, leading to a mean-field theory that extends previous results and renormalizes the ratio $r$. We show that the cycle density directly affects the transition from low to high conductance. Conclusions follow.

\section{Cycle equivalence classes and memristive circuits}

\subsection{Circuits, Kirchhoff's laws, and projector operators}

In circuits, Kirchhoff's laws are manifestations of the conservation of physical quantities such as charge and energy. Mathematically, these can be expressed via the introduction of projection operators \cite{bollobas2012graph, Caravelli2016rl}, i.e. matrices $\boldsymbol{\mOmega}$ satisfying the constraint $\boldsymbol{\mOmega}^2=\boldsymbol{\mOmega}$, and directly connected to circuit topology. For instance, for a resistive circuit made of identical resistances of value $r$ in series with voltage generators in series, Ohm's law for the network can be expressed as
\begin{eqnarray}
    \vec i=\frac{1}{r}\boldsymbol{\mOmega} \vec s,
    \label{eq:rescirc}
\end{eqnarray}
where $\vec s$ is the collection of voltage generators connected in series to each resistance, while $\vec i$ contains the branch currents. The value of the current is connected to the effective graph resistance \cite{effgr}. For a circuit of identical resistors, if one picks two nodes $a,b$ and applies a unital voltage, the effective resistance $R(a,b)$ is equivalent to the inverse of the effective current that flows through the generator. The effective graph resistance is $R_{eff}= \sum_{(a,b)} R_{ab}.$ The graph resistance can be evaluated via the graph Laplacian, defined as the matrix $Q_{ab}=1$ if $a=b$, $-1$ if $(a,b)$ is an edge of the graph and 0 otherwise. Then, if $\vec e_a$ is the vector with all zeros and $1$ in position $a$, we have that
\begin{eqnarray}
    R(a,b)=(\vec e_a-\vec e_b)^t Q^{-1}(\vec e_a-\vec e_b)\label{eq:efflap}
\end{eqnarray}
where $Q^{-1}$ is intended as the pseudo-inverse of $Q$, e.g. the matrix obtained by diagonalizing the matrix, and inverting the diagonal matrix of eigenvalues for only the non-zero eigenvalues. It follows that $R_{eff}= V\sum_{i=2}^V \mu_i^{-1}$, $\mu_i$ are the eigenvalues of $Q$. It is known using the spectrum of the Laplacian that for a complete graph $R_{eff}=V-1$.

\begin{figure}
    \centering
    \includegraphics[scale=0.4]{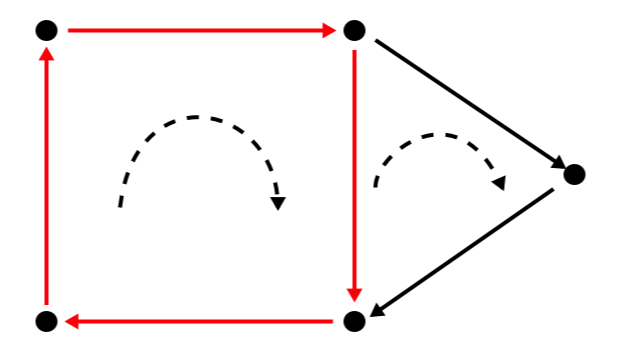}
    \caption{Cycles in graphs. Given an oriented graph, we assign an orientation to the edges and each cycle. Given a spanning tree $\mathcal T$, the number of fundamental cycles are associated to the element not in $\mathcal T$.}
    \label{fig:tikz-egplanargraph_cycle}.
\end{figure}

\subsection{Graph theoretic formalism}
The underlying assumption of \eqref{eq:rescirc} is that the voltage generators $s_i$'s are in series to the resistances $i$'s, while the circuit can be represented as a graph with $E$ edges. Given the branch currents and a certain orientation of the (fundamental) graph cycles $1,\dots, L$, we can obtain the so-called cycle matrix of the circuit $A$, of size $L\times E$, such that $\boldsymbol{\mOmega}=\boldsymbol{A}^t(\boldsymbol{A} \boldsymbol{A}^t)^{-1} \boldsymbol{A}$, where $^t$ denotes the transpose. The fundamental cycles $L$ can be obtained by picking any spanning tree $\mathcal T$, and associating the fundamental cycles in the complement of the tree $\bar{\mathcal T}$.\footnote{Details about this construction can be found in \cite{Caravelli2016rl}}.

Although the details of the derivation of $\boldsymbol{\mOmega}$ from the circuit topology are beyond the scope of this paper, where $\boldsymbol{\mOmega}$ will be kept generic and unrelated to any underlying graph, it is worth stressing its relation to the graph structure of the circuit.
We begin by introducing some concepts from graph theory before defining the projector operator $\mOmega$. A \emph{walk} on a directed graph $G$ is a sequence of vertices, denoted as $v_1, v_2, \dots, v_n$, such that for every pair of consecutive vertices in the sequence, there exists an edge that connects them. Formally, for every $i$ such that $1 < i \leq n$, either $(v_{i-1}, v_i) \in E(G)$ or $(v_i, v_{i-1}) \in E(G)$, where $E(G)$ is the edge set of the graph $G$.

A \emph{cycle} is a walk denoted as $W = v_1 v_2 \dots v_n$, such that $n \geq 3$, $v_0 = v_n$, and the vertices $v_i$, $0 < i < n$, are distinct from each other and from $v_0$. An example of a cycle is shown in Figure \ref{fig:tikz-egplanargraph_cycle}.
The space spanned by the edges of the graph has a vector space structure. The \emph{cycle space} is the subset of the edge space that is spanned by all the cycles of the graph. 

A \emph{cycle matrix} $A$ is a matrix whose columns form the basis of the cycle space. For example, if ${\vec c_1, \dots, \vec c_n}$ is a set of column vectors that form a basis of the cycle space, then the cycle matrix is $A = (\vec c_1, \dots, \vec c_n)$.
Finally, the \emph{projector operator $\mOmega$ on the cycle space of the graph} is defined as $\mOmega = A(A^T A)^{-1}A^T$. The issue with this formalization is that calculating the cycle matrix is cumbersome, as this is a nonlocal quantity. Evaluating $\mOmega$ is easier instead if one uses instead a local quantity, such as the local node connectivity. Let $G$ be a directed graph. One way of representing $G$ is by specifying where each of its edges starts and where it ends. It is convenient to do this using a matrix. We call such a matrix an \emph{incidence matrix}.
Each column of its incidence matrix represents an edge: the first edge starts at vertex $1$ and ends at vertex $2$, so the first column of the matrix has entry $1$ in the first row and entry $-1$ in the second row. All the other entries in the first column are $0$ because none of the other vertices are a part of that edge. By continuing this process for every edge, we get the incidence matrix $B$ of the given graph:

\begin{equation}
B =
\begin{pmatrix}
 1 &-1 &-1 & 0 & 0 & 0 & 0 & 0 \\  
-1 & 0 & 0 & 1 & 1 & 0 & 0 & 0 \\
 0 & 1 & 0 &-1 & 0 &-1 &-1 & 0 \\
 0 & 0 & 0 & 0 & 0 & 1 & 0 &-1 \\
 0 & 0 & 1 & 0 &-1 & 0 & 1 & 1 \\
\end{pmatrix}
\end{equation}

More formally, if a graph $G$ has $v$ vertices and $e$ edges, then the incidence matrix $B$ of $G$ is a $v \times e$ matrix (i.e. a matrix with $v$ rows and $e$ columns), whose entry $(i,j)$ is defined as

\begin{equation}
B_{ij} \coloneqq \begin{cases}
1 & \text{if $v_i$ is the initial vertex of the edge $e_j$,} \\
-1 & \text{if $v_i$ is the terminal vertex of the edge $e_j$,} \\
0 & \text{otherwise.}
\end{cases}
\end{equation}

If $B$ is an incidence matrix, we can define the projector operator $\mOmega_{B^T}$:

\begin{equation}
\label{eq:Omega_def}
\mOmega_{B^T} = B^T\left(B B^T\right)^{-1} B,\ \ \ \mOmega_{B^T}^2 = \mOmega_{B^T}.
\end{equation}

If we try to compute $\mOmega_{B^T}$ from the definition, we will find that the inverse $\left(B B^T\right)^{-1}$ does not exist in general. This can be solved by either considering the reduced incidence matrix (obtained by removing a row from the original incidence matrix) or by taking the pseudoinverse of the expression instead of the ``regular" inverse.

The following useful identity connects the projector operator $\mOmega$ (based on the matrix $A$) to the projector operator $\mOmega_{B^T}$ (based on the incidence matrix $B$):
$\mOmega = \mI - \mOmega_{B^T}$, where $\mI$ is the identity matrix. Such duality between the incidence and cycle adjacency matrix is instead a manifestation of the cycle and nodal analysis in circuit theory.
Such a relationship is based on the fact that $\mOmega$ and $\mOmega_{B^T}$ are orthogonal matrices and the fact that $A B^t=0$, an identity related to the conservation of energy called Tellegen's theorem \cite{Caravelli2017}.

\begin{figure}
    \centering
    \includegraphics[scale=0.25]{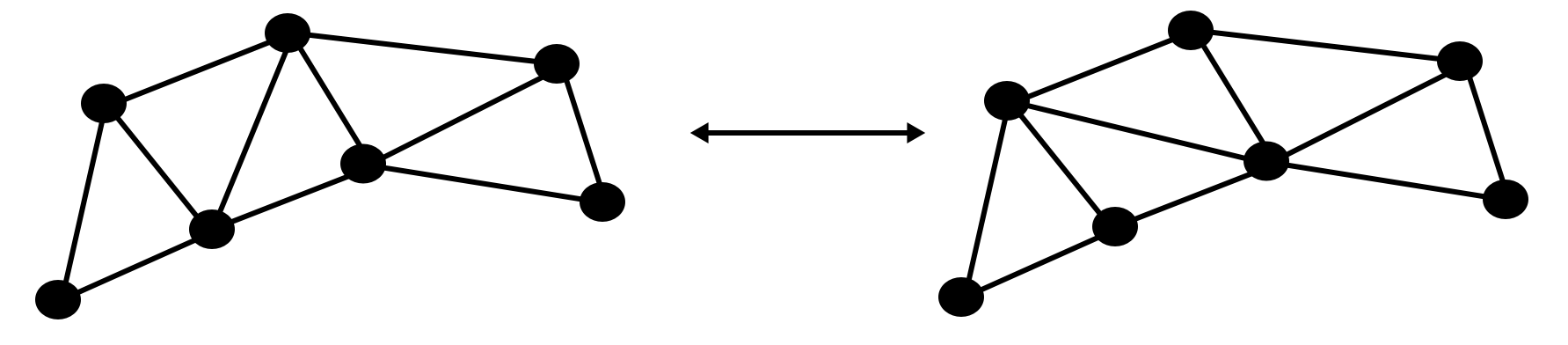}
    \caption{Example of an isospectral transformation for the projector operator $\mathbf \mOmega$. The number of fundamental cycles is preserved, and thus $|Span(A)|$ and the number of edges is preserved.}
    \label{fig:isospectral}
\end{figure}

\subsection{Orthogonal group and cycle preserving transformations}

In essence, the projector $\mOmega$ captures the part of a vector that can be represented as a combination of cycles in the graph, and discards any components that do not lie within the cycle space. Let us then introduce a set of transformations that preserve its spectrum.

\begin{definition}[Cycle preserving transformation]. Let $G=(\mathcal E,\mathcal V)$ be a graph. Let $\theta: G\rightarrow G$ be a map, such that for $\mathcal G \in G$, we have $|\mathcal L( \mathcal G)|=|\mathcal L\big(\theta(\mathcal G)\big)|$ and $|\mathcal E(\mathcal G)|=|\mathcal E\big(\theta(\mathcal G)\big)|.$ Then, $\theta$ is cycle preserving.
\end{definition}

Cycle preserving transformations can be thought as transformations in which an edge is disconnected from the graph, and connected to two other nodes where such edge was not present, as we consider for simplicity simple graphs. An example of such transformation is shown in Fig. \ref{fig:isospectral}. As we can see, the transformation preserves the total number of edges.

Let us now introduce a very simple result, but key to the rest of this manuscript.
The significance of the following lemma lies in its ability to establish a crucial connection between edge-preserving graph transformations and the orthogonal group. Despite its simplicity, this lemma plays a pivotal role in the following, by providing a fundamental insight that forms the basis for our further analyses and conclusions.
\vspace{0.5cm}
\begin{proposition}\label{lemma:iso}Let $\mathcal G$ be a graph $\mathcal G=(\mathcal E,\mathcal V)$.  Let $\theta$ be cycle preserving. If $\mathbf \mOmega_{A(\mathcal G)}$ is the projector operator on $\mathcal L=\text{Span}(A)$, we have
\begin{equation}
    \mathbf \mOmega^\prime\equiv \mathbf \mOmega_{A\big(\theta(\mathcal G)\big)}=\mOt\mathbf \mOmega_{A(\mathcal G)} \mO
\end{equation}
for some orthogonal matrix $\mO$.
\end{proposition}
\begin{proof}. Since $\theta$ preserves the cycle space, we have $L=|\mathcal L|=|\mathcal L^\prime|$ and $|\mathcal E|=|\mathcal E^\prime|$. Then, we know that $\text{Span}(A)=\text{Span}(A^\prime)$, since $\mathbf \mOmega^2=\mathbf \mOmega$, we have that the spectrum is given by $\Lambda(\mathbf \mOmega)=\{0^{E-L},1^{L}\}$ is preserved. As such, $\mathbf \mOmega$ and $\mathbf \mOmega^\prime$ are isospectral, and thus since both $\mathbf \mOmega$ and $\mathbf \mOmega^\prime$ are real,  $\mathbf \mOmega^\prime=\mOt \mathbf \mOmega \mO $ for some orthogonal matrix $\mO$.
\end{proof}
\vspace{0.5cm}
We say that two graphs $\mathcal G$ and $\mathcal G^\prime$ are \textit{conjugate} if $\mOmega$ and $\mOmega^\prime$ are related by an orthogonal transformation. It is easy to see that, in fact, such transformations form an equivalence class. 
Such a simple result  will go a long way in this paper, as it allows us to perform averages over the orthogonal transformations and derive a mean field result in each conjugacy class, e.g. averaging on graphs with the same number of cycles.

\vspace{0.5cm}


To see a direct application of this result, consider again the solution for a resistive circuit, eqn. (\ref{eq:rescirc}).
\begin{figure}
    \centering
    \includegraphics[scale=0.25]{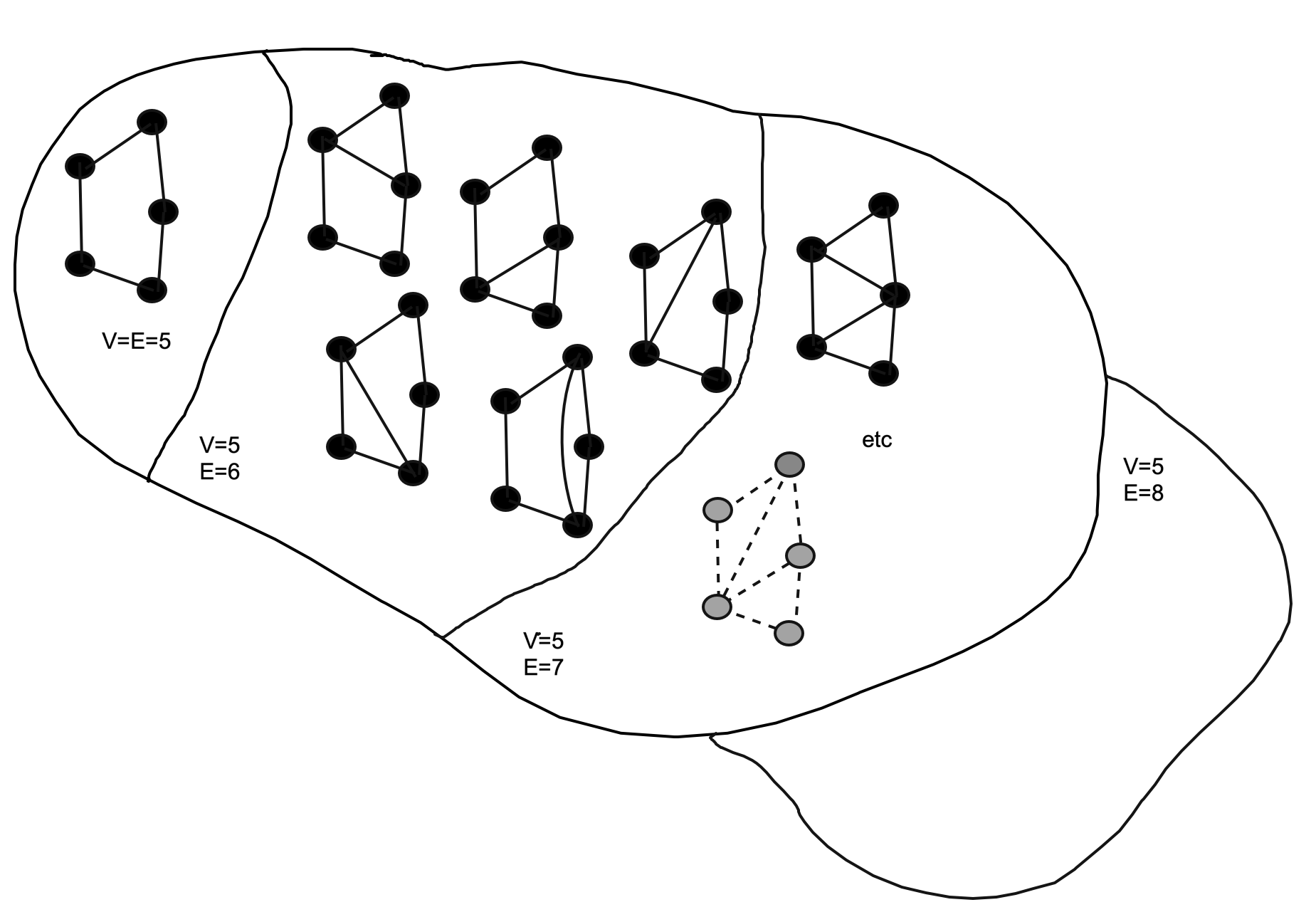}
    \caption{Examples of partitioning in equivalence classes for $V=5$ for graphs representing circuits. As we can see, the elements in each class are related by an edge-preserving transformation. We characterize each equivalence class as $\mathrm{C}(E,L)$, or alternatively $\mathrm{C}(E,V)$.}
    \label{fig:eqclassv5}
\end{figure}
If two graphs representing the resistive circuits are conjugate, then because of Proposition \ref{lemma:iso} we know that
\begin{eqnarray}
    \vec i&=&\mOmega \vec s\\
    \vec i^\prime &=&\mO^t \mOmega \mO \vec s
    \label{eq:currentveco}
\end{eqnarray}
Which implies immediately that
\begin{eqnarray}
    \mO \vec i^\prime =\mOmega \mO \vec s, 
\end{eqnarray}
from which we get that up to a rotation of currents and voltages, the two are identical. This implies that there is a class of circuits whose currents and applied voltages can be mapped to each other via an orthogonal transformation. This is why we can introduce an equivalence class induced by such transformations. We can then classify the number of such equivalence classes.

\begin{definition}
We say that $\mathcal G\sim \mathcal G^\prime$ if $\mOmega(\mathcal G)=\mO^t \mOmega(\mathcal G^\prime) \mO$.
\end{definition}

Using the equivalence relation, we have the following lemma for circuits

\begin{proposition} \label{lemma:eqclasses}
    Let $\mathcal G$ be a graph on $V$ nodes representing a circuit, e.g. there are no nodes of degree 1 and $\mathcal G$ is connected. Then there are at most  $\frac{1}{2} (V^2-3V+2)$ and minimum  $2V-1-\big(V\ \text{mod}\  2)/2$ cycles respectively. This classifies the equivalence classes with fixed nodes.
\end{proposition}

\begin{proof}
 The number of equivalence classes is the number of possible spectral configurations of $\mOmega$. This is given in the most general case, $\Lambda(\mOmega)=\{1^{L},0^{E-L}\}$.
Circuits must satisfy certain relationships between the number of cycles, the number of fundamental faces and the number of vertices. First, note that the ordering of spectrum is not important, as we can permute the eigenvalues via a permutation matrix $P$, and then introduce the orthogonal matrix $\mO\rightarrow P\mO P^t$ to reorder them. Since $P\mO P^t$ is still an orthogonal matrix, permutations of spectra belog to the same class.  Thus, the eigenvalue ordering is not important, but only the number of 1's and 0's, the spectral fingerprint. The number of fundamental cycles is equal to the number $L=E-|\mathcal T|$ where $\mathcal T$ is the cardinality of a spanning tree. Because the graph of a circuit is connected, a spanning tree always exists with the number of edges $|\mathcal T|=V-1$ (e.g. we do not consider forests). 
It is easy to see that the number of leaves is even if $V\ \text{mod}\ 2=0$ and the number of leaves is odd if $V\ \text{mod}\ 2=1$. Let us assume $V=2k$ even or $V=2k+1$ odd with $k\geq 1$. Then, we need to add at least $k$ edges between the leaves to form $k$ cycles. Thus, the minimum number of fundamental cycles is $k$.  We can then add one edge at a time until we obtain a complete graph.  Then, for a graph with $V=2k(+1)$ nodes, we have at most $L=V(V-1)/2 - V+1=\frac{1}{2} (V^2-3V+2)$ (the graph $\mathcal K_V$), and minimum $L=V-1+ V-\big(V\ \text{mod}\  2)/2\big)=2V-1-\big(V\ \text{mod}\  2)/2\big)$ cycles (the graph in which we added the minimum number of edges to remove the leaves from the spanning tree).
\end{proof}

\begin{corollary}
    All graphs with fixed number of edges and vertices $\mathcal G=(\mathcal E,\mathcal V)$ representing a circuit belong to the same equivalence class under $\sim$. We call this equivalence class $\mathrm{C(E,V)}$.
\end{corollary}

\begin{proof}
    This is a consequence of  Proposition \ref{lemma:iso} and Proposition \ref{lemma:eqclasses}. If $V$ is fixed, and $E$ is fixed, then $L=E-V+1$. Then, if $E$ and $V$ are fixed, then $\forall \mathcal G$ such that $|\mathcal E|=E$ and $|\mathcal V|=V$ the number of fundamental cycles is fixed. Then all graphs in the same equivalence class are related to each other via local edge moves as in Fig. \ref{fig:isospectral} which preserve the number of edges and nodes. 
\end{proof}

The previous results establish that in order to remain into the same equivalence class it is sufficient to perform moves that preserve the number of edges, given the number of vertices. Then, it follows that $\mOmega\sim \mOmega^\prime$, then we know the number of edges and the number of vertices must be identical.
An example of such equivalence class partitioning is shown in Fig. \ref{fig:eqclassv5}.

Before we continue, let us make some comments on the relation between number of cycles, edges and nodes. In particular, it will be important to calculate the ration $L/E$. For a connected graph, the cardinality of a spanning tree is equal to $V-1$, and thus
\begin{eqnarray}
    \frac{L}{E}=\frac{E-V+1}{E}
    =1- \frac{V}{E}+\frac{1}{E}.
\end{eqnarray}
Because of the handshaking lemma, we know that $V \bar d= 2 E$, from which we obtain that
$\frac{V}{E}=\frac{2}{\bar d}$, where $\bar d$ is the average degree.
We then have
\begin{eqnarray}
    \frac{L}{E}=
    1- \frac{2}{\bar d}+\frac{1}{E}.
    \label{eq:conntodegree}
\end{eqnarray}
We see from the equation above that for complete graphs this ratio goes to one.

\subsection{Averaging over the equivalence class}
One question we might ask is then: what is the \textit{average} current vector on the equivalence class determined by cycle isomorphism? What we can do is average over all orthogonal matrices which span an equivalence class $\mathrm{C}(E,L)$, and write
\begin{eqnarray}
    \langle \vec i\rangle_\mO=\frac{1}{r} \langle \mOt \mOmega \mO\rangle_{\mO} \vec s.
    \label{eq:averageo}
\end{eqnarray}
We can identify one element in the class $\mathrm{C}(E,V)$, e.g. the graph $\mathcal G_0$ and associate it to the identity element $\mO_0=\mathbb I\in \mathbb O(E)$, the group of orthogonal matrices of size $E$. Then, given $\mOmega_0$ associated to $G_0$, $\forall G_i\in \mathrm{C}(E,V)$, there exist an orthogonal matrix $\mO_{i}$ s.t. $\mOmega_i=\mO_i^t \mOmega_0 \mO_i$.  For generic graphs, such average is over a finite number of elements. As an example, consider Fig. \ref{fig:eqclassv5}. For $V=5, E=5$, there is a unique representative for this equivalence class. For $V=5, E=6$, there are instead 5 representatives. This means that there is, up to permutations of nodes and edges, only a finite number of orthogonal matrices to average over, and thus the average above is over a finite sum. Let $S(E,L)$ be the number of elements in $\mathrm{C}(E,V)$.
Then, our average of eqn. (\ref{eq:averageo}) can be written as
\begin{eqnarray}
    \langle \cdot \rangle_{\mO}=\frac{1}{S(E,L)} \sum_{i=0}^{S(E,L)-1}  \mO_{i}^t \cdot \mO_i
    \label{eq:averageo1}
\end{eqnarray}
and thus we obtain that
\begin{eqnarray}
 \langle \vec i \rangle_{\mO}=\frac{1}{r}\frac{1}{S(E,L)} \sum_{i=0}^{S(E,L)-1}  \mO_{i}^t \mOmega  \mO_i \vec s
    \label{eq:averageo2}
\end{eqnarray}

For finite $V,E$, we are not able to perform such an average in this paper. However, a possibility might be to perform the average when the number of edges and loops goes to infinity. In this case, the number of orthogonal matrices also goes to infinity. A mathematical question might then be what is the measure in the limit of $E,L$ going to infinity? Our working assumption going forward will be that such average converges to the \textit{continuous} group of orthogonal matrices, but this is a rather non-trivial statement. In the following, this will be our working assumption and in the worst-case scenario our approximation. The average over continuous orthogonal matrices can be performed using the Haar measure, e.g. we will perform the replacement
\begin{eqnarray}
    \lim_{E,L\rightarrow \infty}\frac{1}{S(E,L)} \sum_{i=0}^{S(E,L)-1}  \mO_{i}^t \cdot \mO_i\rightarrow \int_{\mathbb O(E)} \mO_{i}^t \cdot \mO_i d\mO, \nonumber 
\end{eqnarray}
where $d\mO$ is the continuous Haar measure over the orthogonal group \cite{Haar}. 
Thus the group $\mathbb O(E)$ is equipped with a Haar probability measure, we can define the average
\begin{eqnarray}
    \langle \cdot \rangle_{\mO}\equiv 
    \int_{\mathbb O(E)} \cdot dO.
\end{eqnarray}
Then, performing this average is equivalent to averaging over random orthogonal matrices with a constant measure over $\mathbb O(E)$.

To motivate this approach, let us consider a very simple example, based on the exact current solution of resistive network, eqn. (\ref{eq:currentveco}).
We have an expression of the form
\begin{eqnarray}
    \langle \mO^t \mOmega \mO\rangle_{\mO}
\end{eqnarray}
where the average is performed via the Haar measure over the orthogonal group. Let us anticipate one result here: it is known that the operation above is the first order \textit{isospectral twirling} of the matrix $\mOmega$, and if one is acquainted with these methods, it is immediate to see that $\langle \mOt \mOmega \mO\rangle_{\mO}\propto \mI$. The reason why this is the case will be clear in the next section, where we describe in detail the techniques used to perform these averages in detail. We will just state for the sake of clarity here that
such average gives
\begin{eqnarray}
    \langle \vec i\rangle_\mO=\frac{\Tr(\mOmega)}{r E} \vec s.
\end{eqnarray}
Above, $E$ is the size of the matrix $\mO$, which in our case is the number of edges of the graph, and thus resistors. The quantity $\Tr(\mOmega)$ represents the matrix trace, and is simply the sum over the eigenvalues of $\mOmega$. We know from the introduction that, because the $\mOmega$ is a projector over the number of fundamental cycles, then the number sum of 1's and 0's in the spectrum is simply $L$. It follows that
\begin{eqnarray}
    \langle \vec i\rangle_\mO=\frac{L}{rE} \vec s.
\end{eqnarray}
We see that the average over rotations of the matrix $\mOmega$ leads to a trace, which contains the information over the network topological properties, such as the number of cycles. The amount above is an average over all possible real projectors with the same spectrum, and it provides information on the expected size of the current as a function of graph topological quantities, such as the number of fundamental cycles and the number of edges/resistors. Thus, this result suggests that our system of currents is equivalent to a set of \textit{uncoupled} resistances of effective value $r^\prime= \frac{r E}{L}$. To see whether we can make sense of this, consider $E$ parallel resistance.
We know from earlier discussions that for complete graphs, such a ratio goes to one. This is the same result obtained from the effective resistance analysis obtained in \cite{effgr} using eqn. (\ref{eq:efflap}). We thus expect our Haar average to be informative for very dense circuits.

The second question one might ask is whether this result is typical. This means asking what is the probability that, given two isospectral graphs, a certain measure between the two vectors is far from a certain value. For instance, let us pick
\begin{eqnarray}
    \text{Var}(\vec i)=\| \vec i-\langle \vec i\rangle_{\mO}\|^2 .
\end{eqnarray}
We can expand the expression above, and write it as
\begin{eqnarray}
    \text{Var}(\vec i)&=&\vec i\cdot \vec i+\langle \vec i\rangle_{\mO}\cdot \langle \vec i\rangle_{\mO}-2 \langle \vec i\rangle_{\mO}\cdot \vec i \\
    &=&\vec s\ ^t \cdot \mOmega \vec s +\frac{L^2}{E^2} \vec s\ ^t \cdot \vec s -2 \frac{L}{E} \vec s\ ^t \cdot \mOmega \vec s 
\end{eqnarray}
where we used $\mOmega^2=\mOmega$. We can average this expression again, and obtain 
\begin{eqnarray}
    \langle \text{Var}(\vec i)\rangle_\mO=\frac{L}{E} \|\vec s\|^2\big(1-\frac{L}{E}\big).
\end{eqnarray}
There are now two possible cases. First, consider $\|\vec s\|^2=O(1)$ e.g. not scaling with $E$. In the limit in which $E\gg L$, or $E=L$ such a result is typical because of concentration inequalities such as Chebyshev inequalities. This is the case when for instance we have a single cycle and a large number of resistors. The second is the case when for instance the graph is extremely dense, for instance in the case of a complete graph, in which $L\propto E$.
Using this technique, we can thus make estimates about the typical behavior of a certain system for very dense graphs. 

While resistive circuits might not be considered as interesting as other physical systems, the point of our paper is that similar techniques can also be used in the case of dynamical resistive circuits, such as memristive circuits. Unfortunately, these are typically more complicated, and as we show in this paper, even for the case of the simplest model of memristive dynamics, exact results can be obtained only in the case of dense graphs and in the asymptotic limit.  The question is then what happens then in the case of  memristive circuits?

\subsection{Toy model of memristive circuit dynamics}

Let us now discuss the explicit equation of memristive dynamics for linear and current controlled devices, as a toy model for the conductance transitions. The equation of motion for a circuit of memristors has been derived in a previous study \cite{Caravelli2017,caravelliscience}. It models a flow network that adheres to current and energy conservation laws, and the dynamics of its edges are bounded within an interval. A memristor with memory can be described by an effective resistance that depends on an internal parameter $x$. For example, $TiO_2$ memristors can be approximated by the functional form $R(x)=R_{on}(1-x)+xR_{off}$, where $R_{on}<R_{off}$ are the limiting resistances, and $x\in[0,1]$ physically represents the size of the oxygen-deficient conducting layer \cite{stru08}. The internal memory parameter $x$ evolves, to the lowest order of description, according to a simple equation:

\begin{equation}
\frac{d}{dt}x=\frac{R_{off}}{\beta} I-\alpha x=\frac{R_{off}}{\beta} \frac{V}{R(x)}-\alpha x,
\end{equation}

with hard boundaries. Here, $\alpha$ and $\beta$ are the decay constant and the effective activation voltage per unit of time, respectively, and they determine the timescales of the dynamical system.
Many extensions of this basic model have been considered in the literature. For example, diffusive effects near the boundaries can be approximated by removing the hard boundaries and multiplying by a window function \cite{Joglekar2009, Biolek2013, Prodromakis2011}. Nonlinear conductive effects can also be included by replacing $I$ with a function $f(x, I)$ or introducing new parameter dependencies, such as temperature in the case of thermistors \cite{Ginoux2020}. Comparisons between these models \cite{Ascoli2013, Corinto2012, Corinto2012a, Ascoli2015} show that many of them are more faithful to the precise IV curves of physical devices, but most share the basic pinched hysteresis phenomenology of the linear model. In analytical work, it is often assumed that the dynamics are linear in the currents in order to study the behavior in a wide context. 

For a single memristor under an applied voltage $S$, Ohm's law $S=RI$ can be used to obtain an equation for $x(t)$ in adimensional units , given by:

\begin{eqnarray}
\frac{d}{dt}x=-\alpha x+\frac{S}{\alpha \beta} \frac{1}{1-\chi x}-x=-\partial_x V(x,s),
\label{eq:oned}
\end{eqnarray}

where $\chi=\frac{R_{off}-R_{on}}{R_{off}}$ and $s=\frac{S}{\alpha \beta}$, with $0\leq \chi \leq 1$ in physically relevant cases, and $V(x, s)$ represents an effective potential.

The dynamics of the one-dimensional dynamical system is described by an ODE,  and is fully characterized by gradient descent in the potential:
\begin{equation}
V(x,s)=\frac{\alpha}{2} x^2 +\frac{s}{\chi} \log(1-\chi x),
\label{eq:potential}
\end{equation}

The equation above, in the case of a circuit of purely memristive devices with a voltage generator in series, is generalized to
a set of coupled ODE written compactly as \cite{caravelliscience}
\begin{eqnarray}
    \frac{d\vec x}{dt}=\frac{1}{\beta}(\mI-\chi \mOmega \mX)^{-1}\mOmega \vec s-\alpha \vec x.\label{eq:dyne}
\end{eqnarray}
We can see that the (\ref{eq:dyne}) is written for arbitrary network topologies, as these enter only via $\mOmega$.

To see why $\mOmega$ is connected to Kirchhoff's laws for memristive devices, consider again eqn. (\ref{eq:dyne}). It is easy to see that the transformation $\vec s\rightarrow \vec s+(I-\mOmega) \vec k$ for arbitrary vectors $\vec k$ leaves the dynamics invariant, since $\mOmega(\mI-\mOmega)=0$. This is a manifestation of Kirchhoff's laws \cite{Caravelli2017}, which are however also present in eqn. (\ref{eq:rescirc}). However,  in memristive devices described by eqn. (\ref{eq:dyne}) this is a subset of a larger invariance set, which are more generally of the form
\begin{eqnarray}
    (\vec x,\vec s)\rightarrow (\vec x\ ^\prime,\vec s\ ^\prime)=(\vec x+\delta \vec x,\vec s+\delta \vec s)
\end{eqnarray}
which preserves the dynamics, e.g. such that  the time derivative is invariant, e.g. $\frac{d\vec x}{dt}=\frac{d\vec x\ ^\prime}{dt}$. To see this, let us equate
\begin{eqnarray}
    (\mI-\chi \mOmega \mX)^{-1}\mOmega \frac{\vec s}{\beta}&-&\alpha \vec x \nonumber \\
    &=&(\mI-\chi \mOmega (\mX+\delta \mX))^{-1}\mOmega \frac{(\vec s+\delta \vec s)}{\beta} \nonumber \\
    &&-\alpha (\vec x+\delta \vec x)
\end{eqnarray}
After a bit of algebra, we get the following generalized relationship:
\begin{eqnarray}
    (\mI-\mOmega)\delta \vec x&=&0\nonumber \\
    \delta \vec s&=&\alpha \beta \big(\mI-\chi (\mX+\delta \mX)\big)\delta \vec x-\chi \delta \mX(\mI-\chi \mOmega \mX)^{-1}\mOmega \vec s\nonumber \\
    &&\ +(\mI-\mOmega)\vec k \label{eq:meminv}
\end{eqnarray}
for any vector $\vec k$. To see the origin of this generalized invariance, consider for instance resistances $R_1,\cdots,R_n$ in series with $k$ generators $s_i$. We have
\begin{eqnarray}
    \frac{\sum_{j=1}^k s_j}{\sum_{j=1}^n R_j}= i.
\end{eqnarray}
To preserve the currents, we can either change $R_j$ and $s_j$ in such a way that the ratio is preserved. For networks of memristors, eqn. (\ref{eq:meminv}) generalizes this simple invariance.

\subsection{Existing results on the mean field theories and limitations}

The property that we are interested in in this paper is the mean-field behavior eqn. (\ref{eq:dyne}). Interestingly, it has been shown that in the laminar regime, the dynamics  of the whole system can be described in terms of a dynamical \textit{scalar} order parameter $\bar x$.

The first result of this type was based on a random matrix theory approximation of the resolvent. Recent studies have provided evidence for statistical regularities in the resolvent of large matrices, suggesting that it can be approximated by an effectively one-dimensional matrix that is universal at the zeroth-order in the limit of weak correlations \cite{BCCV}. This result has broad applicability across various domains of complexity science \cite{BCCV2}. By defining the vector $\vec f$ as $\Omega \vec x$, and introducing the matrix $\tilde A$ given by
\begin{eqnarray}
\tilde A=\begin{pmatrix}
        f_1(\vec x) & \cdots & f_1(\vec x)\\
        f_2(\vec x) & \cdots & f_2(\vec x)\\
        \vdots & \ddots & \vdots\\
        f_N(\vec x) & \cdots & f_N(\vec x)
        \end{pmatrix}= \vec f\otimes \vec 1^t,
\end{eqnarray}
we can establish an approximate relation  given by
\begin{eqnarray}
(I-\chi\mOmega \mX)^{-1}= I+\frac{1}{N} \frac{\chi}{1-\frac{1}{N}\chi\sum_{i=1}^N f_i(\vec x)}\tilde A +O\left(\frac{1}{N}\right),\nonumber
\label{eq:meanfieldexp}
\end{eqnarray}
where $\chi$ is a parameter, $X$ is the matrix of variables $\vec x$, and $N$ is the size of the matrix. This approximation allows us to express the dynamics in terms of the effective one-dimensional variable $x_{cg}$, defined as the coarse-grained average of $\vec f(\vec x)$, and an operator mean $\langle \vec\square \rangle =N^{-1}\sum_{i=1}^N\square_i $. The resulting effective dynamics can be written as
\begin{eqnarray}
\frac{d}{d\tau} \bar x&=&\frac{1}{\alpha \beta}\Big( \langle \Omega \vec S\rangle+\frac{\chi \langle \Omega \vec S \rangle}{1-\chi x_{cg}}x_{cg}\Big)- x_{cg}+\mathcal L(\vec x) \nonumber \\
&=&-\partial_{x_{cg}} V(x_{cg},\chi) +\mathcal L(\vec x)
\label{eq:meanfielddynp}
\end{eqnarray}
where $\alpha$, $\beta$, and $\vec S$ are parameters, and $\mathcal L(\vec x)$ represents an effective force arising from imperfect coarse-graining. Notably, this dynamics resembles that of a single memristor, with the parameter $\frac{s}{\alpha \beta}$ replaced by the mean-field value $s=\frac{\langle \Omega \vec S\rangle}{\alpha \beta}$ at the zeroth-order approximation. The effective potential above develops two minima when the parameter $\bar s$, defined as the mean of $\vec s\beta$ crosses a certain threshold. 

Instead, in \cite{caravelliadp} the experimental potentiation-depression fit of the two-probe conductance was based on a different approximation, which is the  mean-field approximation of the matrix inverse which minimizes the Frobenius norm, e.g. the quantity $\bar x$ such that $\|(I+\chi \mOmega \mathbb X)^{-1}-(I+\chi \mOmega \bar x)^{-1}\|^2$ is minimized, which is a more brute but effective approximation. The circuit connectivity was then shown to reabsorbed into effective parameters that can then be fit a posteriori.

More recently, in \cite{caravelliPEDS} it was shown that memristive circuits described by (\ref{eq:dyne}) belong to a larger class of dynamical systems called PEDS (Projective embedding of dynamical systems) where fixed points of larger systems can be mapped to fixed points of a lower dynamical systems, using the relationship $\mOmega^2=\mOmega$. However, also there, exact results could be obtained only for very specific types of projector operators which are essentially a slight generalization of a mean-field coarse graining operator. 

We see then that in both cases the topological properties of the network become less important. This should feel to be dissatisfactory, as we do expect that the topological properties of the network \textit{should} matter. Thus, in what follows, we will average while preserving the global properties of the network, that in our case means preserving the spectral properties of the matrix $\mOmega$.

\vspace{1cm}

\section{Haar measure and averages}

\subsection{von Neumann series and its tensor representation}
The key difference between the purely resistive and the purely memristive case is that the dynamics of the purely memristive case contains an infinite sum and terms of the form $\mO^t \mOmega \mO^t \mX$.
To simplify the calculations that follow, we note that the memristive equation can be written in terms of the adimensional time variable $\tau=\alpha t$ and $\mO$, 
and we perform the following change of variable: $\vec x= e^{-\tau} \vec g$. This simply removes the term $-\alpha \vec x$ in the equation, and allows us to work temporarily only with the infinite sum. We now write the von Neumann series for the matrix inverse, in powers of $\chi$, as 
\begin{eqnarray}
    \frac{d \vec g}{d\tau}&=& \sum_{k=0}^\infty  \chi^k e^{k \tau} (\mOt \mathbf \mOmega \mO  \mathbb G)^k \vec s.
\end{eqnarray}
It is thus immediate to see that performing the average over the orthogonal group as we set to do at the beginning of this manuscript is slightly more complicated than the case of the resistive average.

Let us now perform the following matrix manipulations, using the properties of tensor products, which we introduce in \ref{app:tensprod}. These techniques, used commonly in quantum information, are borrowed from the theory of linear algebra when applied to tensor products of linear operators. Let us assume here that $A_i$'s are linear operators on $\mathbb R^E$. Since we are dealing with matrices of the form  $(\mO^t \mOmega \mO^t \mX)^k$, let us rewrite this matrix power in a way in which we can perform the average using well known results.

Let us begin with the simplest non-trivial case, referring to App. \ref{app:tensprod} for the introduction to the tensor product. We have
\begin{equation}
    A^2=\Tr_2((A \otimes A) \mathbb S)
\end{equation}
Above, the trace is partial, and over the second tensor space. The matrix $\mathbb S$ is the swap operator. 

It follows that we can write 
\begin{eqnarray}
    A^k&=&\Tr_2 \big((A^{k-1} \otimes A) \mathbb S_{12}\big)\nonumber \\
    &=& \Tr_{23} ((A^{k-2}\otimes A) \otimes A) \mathbb S_{12}\mathbb S_{23}) \nonumber \\
    &\vdots& \nonumber \\
    &=& \Tr_{2\cdots k} \big((\underbrace{A \otimes \cdots \otimes A}_{k\ times}) \mathbb S_{12} \mathbb S_{23}\cdots \mathbb S_{k-1,k}\big)\label{eq:tensid}
\end{eqnarray}
where above, with an abuse of notation, we wrote at each step $\mathbb S_{ii+1}$ intending that it swaps only the ith and i-$th+1$ line, and leaves intact the other lines. Let us call $\mathbb S_{2\cdots k}=\prod_{i=1}^{k-1}\mathbb S_{i,i+1}$. Graphically, the identity above can be seen in Fig. \ref{fig:at4}, with the swap operators acting on each tensor index.
\begin{figure}
    \centering
    \includegraphics[scale=0.32]{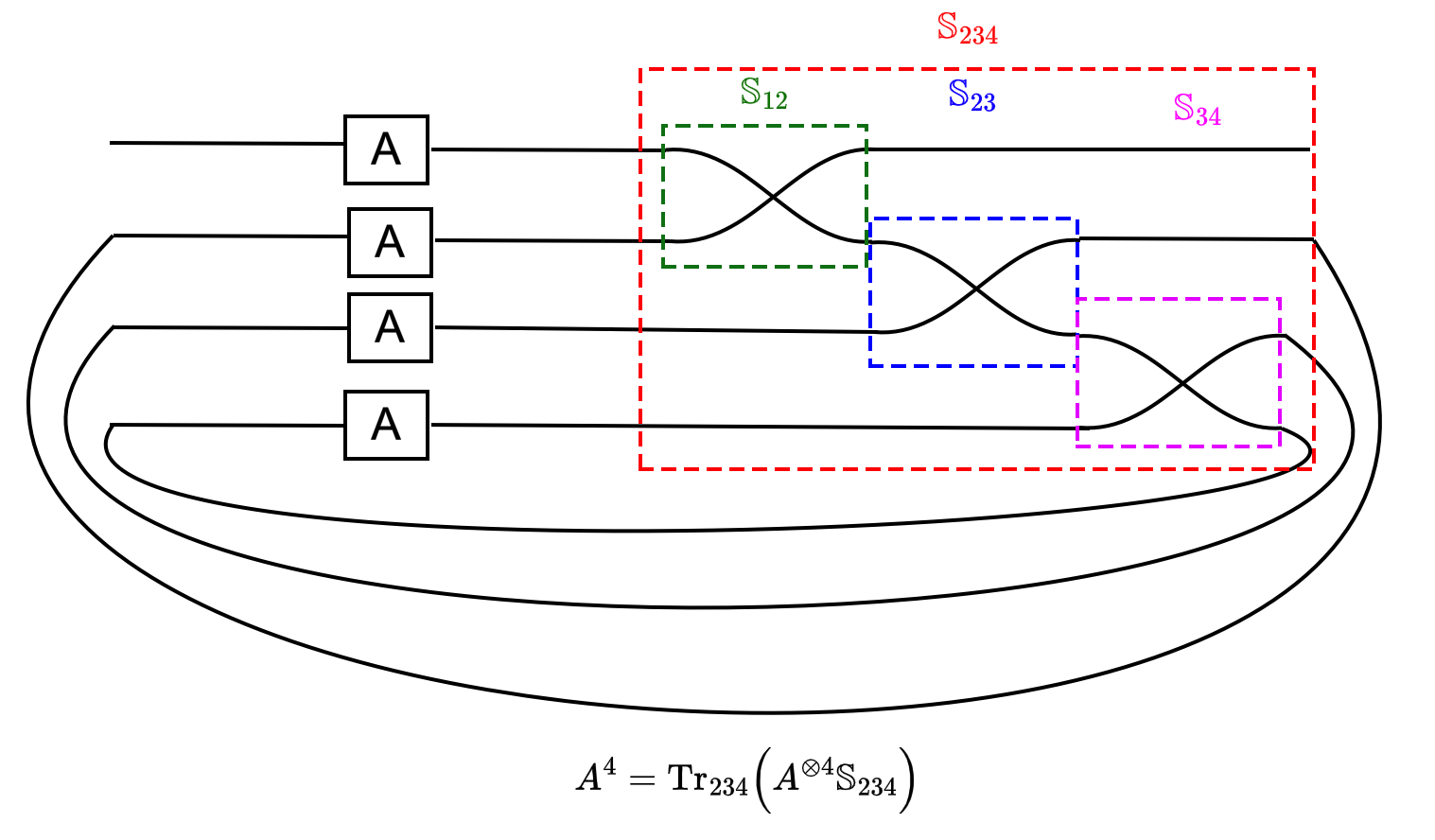}
    \caption{Graphical representation of the identity of eqn. (\ref{eq:tensid}) for $k=4$.}
    \label{fig:at4}
\end{figure}
We now note that, then
\begin{eqnarray}
    (\mOt\mathbf \mOmega \mO \mathbb G)^k&=& \Tr_{2\cdots k} \big((\underbrace{\mOt\mathbf \mOmega \mO \mathbb G \otimes \cdots \otimes \mOt\mathbf \mOmega \mO \mathbb G}_{k\ times}) \mathbb S_{2\cdots k}\big).\nonumber 
\end{eqnarray}

Let us call $(\mOt)^{\otimes k}=\mOt\otimes \cdots \otimes \mOt$ and similarly for $\mathbf \mOmega^{\otimes k}$, $\mO^{\otimes k}$ and $\mathbb G^{\otimes k}$. We have then that
\begin{eqnarray}
    (\mOt\mathbf \mOmega \mO \mathbb G)^k&=& \Tr_{2\cdots k} \big( (\mOt)^{\otimes k} \mathbf \mOmega^{\otimes k} \mO^{\otimes k} \mathbb G^{\otimes k}\mathbb S_{2\cdots k}\big).
\end{eqnarray}
Using this formalism, we then have
\begin{equation}
    \frac{d \vec g}{d\tau}= \sum_{k=0}^\infty  \chi^ke^{k \tau} \Tr_{2\cdots k} \big( (\mOt)^{\otimes k} \mathbf \mOmega^{\otimes k} \mO^{\otimes k} \mathbb G^{\otimes k}\mathbb S_{2\cdots k}\big) \vec s .
\end{equation}
The equation above is written in the form of a twirl. A twirl is an expression of the form \cite{isotwirl}
\begin{eqnarray}
    \langle (\mOt)^{\otimes k} \mathbf \mOmega^{\otimes k} \mO^{\otimes k}\rangle_{\mO}=\mathbb M_k(\mOmega)\label{eq:averagemk}
\end{eqnarray}
where $\mathbb M_k$ is a certain linear operator on $(\mathbb R^E)^{\otimes k}$ following the average. In tensorial graphical form, the expression is shown in Fig. \ref{fig:isotwirl}. The internal index contractions will be shown as blue lines, and external indices to the twirl shown in red.

We now ask for the following question, e.g. what is the average behavior of the circuit when we take the isospectral twirling average with respect to orthogonal operator $\mO$? The techniques used to perform these averages are well established \cite{collins2006integration,zinnjustin,collins2017weingarten,banica}.

We have then that
\begin{equation}
    \big\langle \frac{d \vec g}{d\tau}\big\rangle_{\mO}= \sum_{k=0}^\infty \chi^k e^{k \tau}  \Tr_{2\cdots k} \big( \mathbb M_k \mathbb G^{\otimes k}\mathbb S_{2\cdots k}\big) \vec s .
\end{equation}
The goal will be now to calculate the operators $\mathbb M_k$.

\begin{figure}
    \centering
    \includegraphics[scale=0.38]{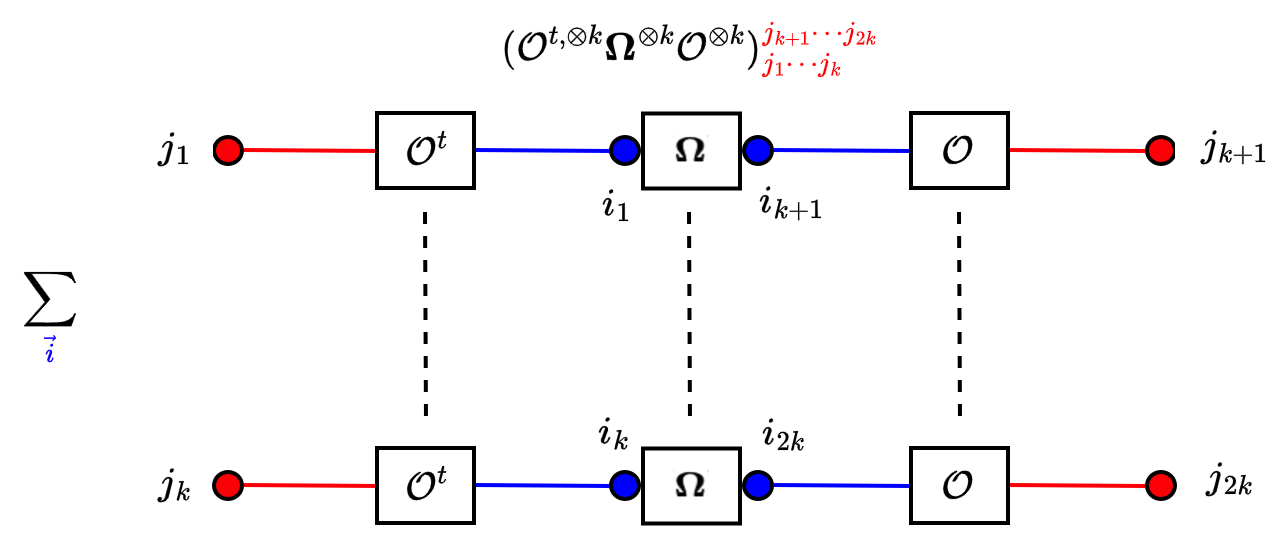}
    \caption{Example of a twirl average in tensorial graphic representation. The operator $\mOmega^{\otimes k}$ is squeezed in between $(\mOt)^{\otimes k}$ and $\mO^{\otimes k}$ which is being averaged over.}
    \label{fig:isotwirl}
\end{figure}
\subsection{Haar measure, Weingarten calculus and mean field theory}

In the present manuscript, we are concerned with performing averages of the form
\begin{eqnarray}
    \langle \cdot\rangle_{\mO}
\end{eqnarray}
over orthogonal matrices.
Let $\mO$ be a $E\times E$ orthogonal matrix, e.g. such that $\mOt=\mO^{-1}$. The set of orthogonal matrices forms a group, which we call $\mathbb O(E)$. For any orthogonal matrix $\mathscr O$, there exists an inverse $\mathscr O^{-1}=\mathscr O^t$ such that $\mO \mOt=\mOt \mO=\mI$. In fact, if $\mO_1,\mO_2$  are orthogonal, it is easy to see that $\mO_1\mO_2$ is also orthogonal, and the identity matrix is associated to the identity in the group. This is a subspace of $\mathcal M_N(\mathbb C)$, the algebra of $E\times E$ complex matrices with operation given by the standard matrix multiplication, and which is a compact topological space.  Haar's theorem \cite{Haar} states that for any locally compact Hausdorff topological group, there exists a left-translation-invariant measure and a right-translation-invariant measure that is unique up to a positive multiplicative constant. In the case of a compact group, these measures coincide and are known as the Haar measure. 
Let $G$ be a group. A representation $R$ of $\mO$ on $\HC$ is defined as a group homomorphism $R: \mO \to GL(\HC)$~\cite{fultonharris}. Given $R$, we can construct the $k$-fold tensor representation of $G$ on $\HC^{\otimes k}$ by acting with $R(g)^{\otimes k}$ for any $g \in G$. It can be observed that if $R$ is a valid representation, then its $k$-fold tensor product also forms a representation. The Haar average over the k-fold tensor product can be performed using the procedure of Weingarten calculus \cite{weingarten78asym, collins2003moments, collins2006integration}, which we will briefly review here.
Consider a bounded operator $A$ on the Hilbert space $\mathcal{H}$, and let $\mO$ be an orthogonal operator chosen uniformly at random from the orthogonal group $\mathbb{O}(E)$. 
We are interested in particular in the \textit{isospectral twirling}. This is the following operator.
The $2k$-Isospectral twirling  of $X$ is defined as \cite{isotwirl}:
\begin{eqnarray}    
  \hat{\mathcal{R}}_{\mathbb O}^{(2k)}(X):=\int\, d \mO\, \mO^{\dag \otimes 2k}\left(X^{\otimes 2k}\right)\mO^{\otimes 2k}
\label{isospectraltwirling}
\end{eqnarray}
Such an operator arises from the linearization of the expectation values of polynomials of $\mO ^{t}X \mO$. 
We see immediately that the operator $\mathbb M_k(\mOmega)$ in eqn. (\ref{eq:averagemk}) is in the form of an isospectral twirling.
The  Haar average will be computed by the Weingarten functions method  for the orthogonal group. 
The Haar measure has the properties $\int d\mO=1$ and $d\mO= d(\mO \mO^\prime)=d(\mO^\prime \mO)$ for any orthogonal matrix $\mO^\prime\in {\mathbb O}(E)$, which are referred to as the left(right)-invariance of the Haar measure.
The general formula to compute the Haar average is given by \cite{collins2006integration}, and for the orthogonal group reads:
\begin{equation}
\langle \mO^{t,\otimes k}A^{\otimes k} \mO^{\otimes k}\rangle_\mO=\sum_{\sigma,\tau\in \mathcal P_{2k}}W^O_g(\sigma\tau^{-1},E)\Tr(A^{\otimes k}\mathbb B_{\sigma})\mathbb B_{\pi}
\label{eq:weingartenmethod}
\end{equation}
where $W^O_g(\sigma\tau^{-1},E)$ denotes the Weingarten function of the group $\mathbb O(E)$.  For the case of the orthogonal group $\mathbb B_{\sigma}$ and $\mathbb B_{\tau}$ are operators associated with elements of the Brauer algebra $\sigma$ and $\tau$, respectively, and we will describe them in a moment. Eqn. (\ref{eq:weingartenmethod}) will be motivated in a moment by showing explicitly the contractions, and the explicit expression for the Weingarten functions in terms of the Brauer elements further below.

The Brauer algebra is a mathematical structure that was introduced by Richard Brauer in the 1930s, and it has connections to both algebra and representation theory, and is denoted here by $\mathcal P_{2k}$.
One interesting feature of the Brauer algebra is its connection to pairings in set theory. In particular, the Brauer algebra can be used to study pairings of elements from two sets in a meaningful way. This connection arises from the fact that the Brauer algebra can be used to represent certain operations on partitions, which are combinatorial objects that encode the ways in which a set can be divided into smaller subsets.
Specifically, the Brauer algebra can be used to describe pairings of elements from two sets when the sets have a certain kind of symmetry. In set theory, a pairing is a function or a relation that maps pairs of elements from two sets to a third set. The Brauer algebra provides a way to describe and analyze these pairings in a systematic manner. Examples of Bauer pairings between $2k$ elements are shown in Fig. \ref{fig:brauer}.

\begin{figure}
    \centering
    \includegraphics[scale=0.09]{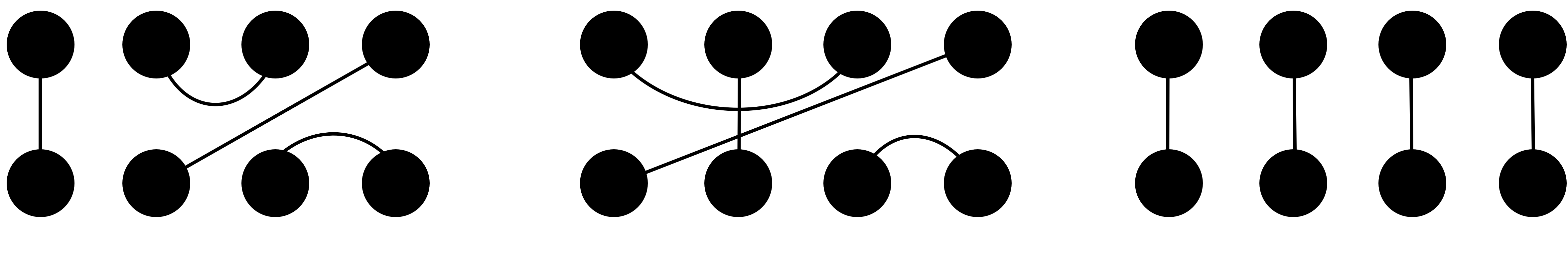}
    \caption{Pairings over 8 elements. Given a set of $2k$ elements, the set of pairings $\mathcal P_{2k}$ is a set of the form $\tau =\{\{i_1,i_2\},\cdots,\{i_{2k-1},i_{2k}\}\}$, with $i_a\neq i_b$. The set can be represented graphically as in the three pairings above. It is common to divide the set in two set of indices, $A_k=\{i_1,\cdots,i_k\},B_k=\{i_{k+1},\cdots, i_{2k}\}$. }
    \label{fig:brauer}
\end{figure}

In the following, we will introduce a coloring scheme to understand better this average. As we can see in Fig. \ref{fig:isotwirl}, the twirling involves external and internal indices, the latter being contrated with between $\mO^t$ $\mOmega$ and $\mO$. We will denote internal indices with a blue color, and external indices with a red color. 
Another way of expressing eqn. (\ref{eq:weingartenmethod}) is by making the indices explicit. 
The average over the orthogonal group is given by the formula
\begin{eqnarray}
    \langle \prod_{r=1}^{2k} \mO_{\textcolor{blue}{i_r}\textcolor{red}{j_r}}\rangle_{\mO}=\sum_{\textcolor{blue}{\sigma},\textcolor{red}{\tau}\in \mathcal P_{2k}} \Delta_{\textcolor{blue}{\sigma}} (\textcolor{blue}{\vec i})\Delta_{\textcolor{red}{\tau}} (\textcolor{red}{\vec j})W_g^O(\textcolor{blue}{\sigma}\textcolor{red}{\tau^{-1}},E)\nonumber 
\end{eqnarray}
Above, we have the following definition:
\begin{eqnarray}
\Delta_{\sigma}(\vec i)=\prod_{\{a,b\}\in \sigma} \delta_{i_a,i_b}.    
\end{eqnarray}
The function above is an index explicit definition of the operators $\mathbb B_{\sigma}'s$ involved in eqn. (\ref{eq:weingartenmethod}).
Using the formula above, the isospectral twirling is written in the form
\begin{widetext}
\begin{eqnarray}
    (\hat{\mathcal{R}}_{\mathbb O}^{(2k)})_{\textcolor{red}{j_1\cdots j_k}}^{\textcolor{red}{j_{k+1}\cdots j_{2k}}}(\mOmega)&=&\sum_{\textcolor{blue}{i_1,\cdots,i_{2k}}} \langle \mO_{\textcolor{blue}{i_1}\textcolor{red}{j_1}}\cdots \mO_{\textcolor{blue}{i_k}\textcolor{red}{j_k}} \mOmega_{\textcolor{blue}{i_1}\textcolor{blue}{i_{k+1}}}\cdots \mOmega_{\textcolor{blue}{i_k}\textcolor{blue}{i_{2k}}} \mO_{\textcolor{blue}{i_1}\textcolor{red}{j_{k+1}}}\cdots \mO_{\textcolor{blue}{i_{k}}\textcolor{red}{j_{2k}}} \rangle_{\mO}\nonumber \\
    &=&\sum_{\textcolor{blue}{i_1,\cdots,i_{2k}}} \langle \mO_{\textcolor{blue}{i_1}\textcolor{red}{j_1}}\cdots \mO_{\textcolor{blue}{i_k}\textcolor{red}{j_k}} \mO_{\textcolor{red}{i_1}\textcolor{blue}{j_{k+1}}}\cdots \mO_{\textcolor{red}{i_{k}}\textcolor{blue}{j_{2k}}} \rangle_{\mO} \mOmega_{\textcolor{red}{i_1}\textcolor{red}{i_{k+1}}}\cdots \mOmega_{\textcolor{red}{i_k}\textcolor{red}{i_{2k}}}\nonumber \\
    &=&\sum_{\textcolor{blue}{i_1,\cdots,i_{2k}}} \sum_{\textcolor{blue}{\sigma},\textcolor{red}{\tau}\in \mathcal P_{2k}} \Delta_{\textcolor{blue}{\sigma}} (\textcolor{blue}{\vec i})\Delta_{\textcolor{red}{\tau}} (\textcolor{red}{\vec j})W_g^O(\textcolor{blue}{\sigma}\textcolor{red}{\tau^{-1}},E)\mOmega_{\textcolor{blue}{i_1}\textcolor{blue}{i_{k+1}}}\cdots \mOmega_{\textcolor{blue}{i_k}\textcolor{blue}{i_{2k}}}\nonumber \\
    &=& \sum_{\textcolor{red}{\tau}\in \mathcal P_{2k}} \Delta_{\textcolor{red}{\tau}} (\textcolor{red}{\vec j})\sum_{\textcolor{blue}{i_1,\cdots,i_{2k}}}\sum_{\textcolor{blue}{\sigma}\in \mathcal P_{2k}}\Delta_{\textcolor{blue}{\sigma}} (\textcolor{blue}{\vec i})W_g^O(\textcolor{blue}{\sigma}\textcolor{red}{\tau^{-1}},E)\mOmega_{\textcolor{blue}{i_1}\textcolor{blue}{i_{k+1}}}\cdots \mOmega_{\textcolor{blue}{i_k}\textcolor{blue}{i_{2k}}}\nonumber \\
     &=& \sum_{\textcolor{red}{\tau}\in \mathcal P_{2k}} \Delta_{\textcolor{red}{\tau}} (\textcolor{red}{\vec j}) R_E (\textcolor{red}{\tau},\mOmega)
\end{eqnarray}
\end{widetext}
where we defined 
\begin{eqnarray}
    R_{E} (\textcolor{red}{\tau},\mOmega)&=&\sum_{\textcolor{blue}{i_1,\cdots,i_{2k}}}\sum_{\textcolor{blue}{\sigma}\in \mathcal P_{2k}}\Delta_{\textcolor{blue}{\sigma}} (\textcolor{blue}{\vec i})W_g^O(\textcolor{blue}{\sigma}\textcolor{red}{\tau^{-1}},E)\cdot\nonumber \\
    && \hspace{2cm} \cdot\mOmega_{\textcolor{blue}{i_1}\textcolor{blue}{i_{k+1}}}\cdots \mOmega_{\textcolor{blue}{i_k}\textcolor{blue}{i_{2k}}}.
    \label{eqapp:isotwird}
\end{eqnarray}
We see then that one of the summation over the Brauer's algebra elements is associated with external indices (red, $\tau$), while the other is associated with internally contracted indices (blue, $\sigma$). The Weingarten function depends instead only on the element $\alpha=\sigma\tau^{-1}$. The element $\alpha$ is essentially the \textit{superposition} of $\sigma$ and $\tau$, for instance, in Fig. \ref{fig:brauer}, we simply overlay two Brauer elements on top of each other. To obtain the Weingarten matrix elements, we first construct the Gram matrix $Gm(\sigma,\tau)$ as follows \cite{zinnjustin,collins2006integration}. At the order $2k$, let $m$ be the number of Brauer' pairings. The Gram matrix is $m\times m$, and in each corresponding matrix elements associated to $\sigma,\tau$ we have $E^{c(\sigma,\tau)}$, where $c(\sigma,\tau)$ is the number of closed cycles obtained by overlaying $\sigma$ and $\tau$ on the same diagram. Then, the Weingarten matrix $W_g^O(\sigma\tau^{-1},E)$ is the corresponding element of the pseudo-inverse of the matrix $Gm(\sigma,\tau)$.

We now have all the elements to perform the average. We need to see the effect of the two summations and the delta function on the indices of the product of $\mOmega$'s.
Let us first consider the case $k=1$. In this case, there are only two elements, and thus the only available pairing is $\textcolor{blue}{\sigma},\textcolor{red}{\tau}=\{1,2\}$.

Since there is only one component, we have $W_g^O(\sigma,E)=\frac{1}{E}$, and thus
\begin{eqnarray}
    R_E(\textcolor{red}{\tau},\mOmega)&&=\sum_{\textcolor{blue}{i_1,i_2}} \frac{\delta_{\textcolor{blue}{i_1,i_2}}}{E} \mOmega_{\textcolor{blue}{i_1 i_2}}=\frac{1}{E} \sum_{\textcolor{blue}{i_1}} \mOmega_{\textcolor{blue}{i_1i_1}}=\frac{\Tr(\mOmega)}{E}\nonumber 
\end{eqnarray}
Thus
\begin{eqnarray}
(\hat{\mathcal{R}}_{\mathbb O}^{(2)})_{\textcolor{red}{j_1}}^{\textcolor{red}{j_{2}}}(\mOmega)=\delta_{\textcolor{red}{j_1 j_2}}\frac{\Tr(\mOmega)}{E}.\label{eqapp:firstorder}
\end{eqnarray}
This is shown in Fig. \ref{fig:IsoTwirlk1k2} (top).

Let us now focus on the case $k=2$. This is already not as simple as the case of $k=1$. The pairing on $4$ elements are now 3: $\sigma_1=\{\{1,3\},\{2,4\}\}$, $\sigma_2=\{\{1,4\},\{2,3\}\}$, $\sigma_3=\{\{1,2\},\{3,4\}\}$. Thus, as we see in Fig.  \ref{fig:IsoTwirlk1k2} (bottom), the average is projected over three operators, and we have
\begin{eqnarray}
    (\hat{\mathcal{R}}_{\mathbb O}^{(4)})_{\textcolor{red}{j_1 j_2}}^{\textcolor{red}{j_{3}j_{4}}}(\mOmega)&=&R_E(\textcolor{red}{\tau_1},\mOmega) \delta_{\textcolor{red}{j_1,j_3}}\delta_{\textcolor{red}{j_2,j_4}} \nonumber \\
    &&+R_E(\textcolor{red}{\tau_2},\mOmega) \delta_{\textcolor{red}{j_1,j_4}}\delta_{\textcolor{red}{j_2,j_3}} \nonumber \\
    &&+R_E(\textcolor{red}{\tau_3},\mOmega) \delta_{\textcolor{red}{j_1,j_2}}\delta_{\textcolor{red}{j_3,j_4}}
\end{eqnarray}
where we now  need to calculate
\begin{eqnarray}
    R_E(\textcolor{red}{\tau_j},\mOmega)=\sum_{\textcolor{blue}{i_1,\cdots,i_{4}}}\sum_{\textcolor{blue}{\sigma}\in \mathcal P_{2k}}\Delta_{\textcolor{blue}{\sigma}} (\textcolor{blue}{\vec i})W_g^O(\textcolor{blue}{\sigma}\textcolor{red}{\tau_j^{-1}},E)\mOmega_{\textcolor{blue}{i_1}\textcolor{blue}{i_{3}}}\mOmega_{\textcolor{blue}{i_2}\textcolor{blue}{i_{4}}}
\end{eqnarray}
As we have seen in the previous example, the quantity $\sum_{\textcolor{blue}{i_1,\cdots,i_{4}}}\Delta_{\textcolor{blue}{\sigma}} (\textcolor{blue}{\vec i})$ results in contractions over the indices $\textcolor{blue}{i_j}'s$, and the rest are summed over.  These thus become immediate that these are traces and products of traces of $\mOmega$.
 The contractions are shown in Fig. \ref{fig:IsoTwirlk2Om}. These corresponds to the elements $\textcolor{blue}{\sigma_1}\rightarrow \Tr(\mOmega)^2$, $\textcolor{blue}{\sigma_2}\rightarrow \Tr(\mOmega^2)$, $\textcolor{blue}{\sigma_3}\rightarrow \Tr(\mOmega^2)$, where we used the fact that $\mOmega$ is a symmetric matrix. However, these must be multiplied by the orthogonal Weingarten functions, which are now non-trivial. The table of the Weingarten coefficients can be obtained by composing $\textcolor{red}{\sigma}$ and $\textcolor{blue}{\tau}$. This is simplified by the observation that $\tau^{-1}=\tau$ since these are pairings. 
 The orthogonal Weingarten function is the pseudo-inverse of the Gram matrix, whose elements $Gm(\textcolor{red}{\sigma},\textcolor{blue}{\tau})=E^{c(\textcolor{red}{\sigma},\textcolor{blue}{\tau})}$, where $c(\textcolor{red}{\sigma},\textcolor{blue}{\tau})$ is the number of connected components of the graph resulting from the composition of $\textcolor{red}{\sigma},\textcolor{blue}{\tau}$. The number of connected components is given, inspecting Figure \ref{fig:wgo2}, by the Gram matrix below
\begin{eqnarray}
    Gm(\sigma,\tau)=\left(
\begin{array}{ccc}
 E^2 & E & E \\
 E & E^2 & E \\
 E & E & E^2 \\
\end{array}
\right)
\end{eqnarray}
We thus find that, ordering the columns and rows according to $\tau_1,\tau_2,\tau_3$ described above, the inverse of the Gram matrix yields the following Weingarten matrix for $k=2$.
\begin{eqnarray}
    &&W_g^O(\textcolor{blue}{\sigma}\textcolor{red}{\tau}^{-1},E)\nonumber \\
    &&\ \ \ \ =\left(
\begin{array}{ccc}
 \frac{E+1}{(E-1) E (E+2)} & -\frac{1}{(E-1) E (E+2)} & -\frac{1}{(E-1) E (E+2)} \\
 -\frac{1}{(E-1) E (E+2)} & \frac{E+1}{(E-1) E (E+2)} & -\frac{1}{(E-1) E (E+2)} \\
 -\frac{1}{(E-1) E (E+2)} & -\frac{1}{(E-1) E (E+2)} & \frac{E+1}{(E-1) E (E+2)} \\
\end{array}
\right)\nonumber 
\end{eqnarray}
The Weingarten matrix is a symmetric matrix with identical elements on the diagonal, and for $k=2$ identical elements on the off-diagonal.
Let us call $C_1=\frac{E+1}{(E-1) E (E+2)}$ the element on the diagonal, and $C_2=-\frac{1}{ E (E-1)(E+2)}$ the elements off-diagonal.

\begin{eqnarray}
    W_g^O=\begin{pmatrix}
        C_1 & C_2 & C_2 \\
        C_2 & C_1 & C_2 \\
        C_2 & C_2 & C_1
    \end{pmatrix}
\end{eqnarray}
Then,  we have
\begin{eqnarray}
&&\ \ \ \  \ \textcolor{blue}{\sigma_1} \hspace{1.3cm}\textcolor{blue}{\sigma_2} \hspace{1.5cm}\textcolor{blue}{\sigma_3} \nonumber \\
    R_1\equiv R_E(\textcolor{red}{\tau_1},\mOmega)&=& C_1  \Tr(\mOmega)^2+ C_2 \Tr(\mOmega^2)+ C_2 \Tr(\mOmega^2) \nonumber \\
R_2\equiv R_E(\textcolor{red}{\tau_2},\mOmega)&=&C_2  \Tr(\mOmega)^2+ C_1 \Tr(\mOmega^2) +C_2   \Tr(\mOmega^2) \nonumber \\
    R_3\equiv R_E(\textcolor{red}{\tau_3},\mOmega)&=&C_2  \Tr(\mOmega)^2+ C_2 \Tr(\mOmega^2) + C_1 \Tr(\mOmega^2)\nonumber 
\end{eqnarray}
from which we
\begin{eqnarray}
    (\hat{\mathcal{R}}_{\mathbb O}^{(4)})_{\textcolor{red}{j_1 j_2}}^{\textcolor{red}{j_{3}j_{4}}}(\mOmega)&=&(C_1  \Tr(\mOmega)^2+ 2C_2 \Tr(\mOmega^2)) \delta_{\textcolor{red}{j_1,j_3}}\delta_{\textcolor{red}{j_2,j_4}}\nonumber \\
    &&\hspace{-0.1cm}+\big((C_1+C_2)\Tr(\mOmega^2)+C_2 \Tr(\mOmega)^2\big)\delta_{\textcolor{red}{j_1,j_4}}\delta_{\textcolor{red}{j_2,j_3}}\nonumber \\
    &&\hspace{-0.1cm}+\big((C_1+C_2)\Tr(\mOmega^2)+C_2 \Tr(\mOmega)^2\big)\delta_{\textcolor{red}{j_1,j_2}}\delta_{\textcolor{red}{j_3,j_4}}\nonumber \\
    \label{eqapp:secondorder}
\end{eqnarray}
Note that we used here the notation in which input indices are at the bottom, while output indices are at the top. We note that $\delta_{\textcolor{red}{j_1,j_3}}\delta_{\textcolor{red}{j_2,j_4}}\equiv \mI^{\otimes 2}$, $\delta_{\textcolor{red}{j_1,j_4}}\equiv \mathbb S$ while $\delta_{\textcolor{red}{j_1,j_2}}\delta_{\textcolor{red}{j_3,j_4}}\equiv  \Pi$, and information theory $\Pi$ is proportional to the projector onto the maximally entangled Bell state between the two copies of a Hilbert space $\mathcal H$.
In the bra-ket notation, if one defines $|\Phi^+\rangle$ as  the Bell state between the two Hilbert spaces, 
\begin{eqnarray}
    |\Phi^+\rangle=\frac{1}{\sqrt{E}}\sum_{j_1=0,j_2=0}^{E-1} |j_1,j_1\rangle\langle j_2,j_2|, 
\end{eqnarray}
then one has $\Pi= E |\Phi^+\rangle\langle \Phi^+|$.

\begin{figure}
    \centering
    \includegraphics[scale=0.24]{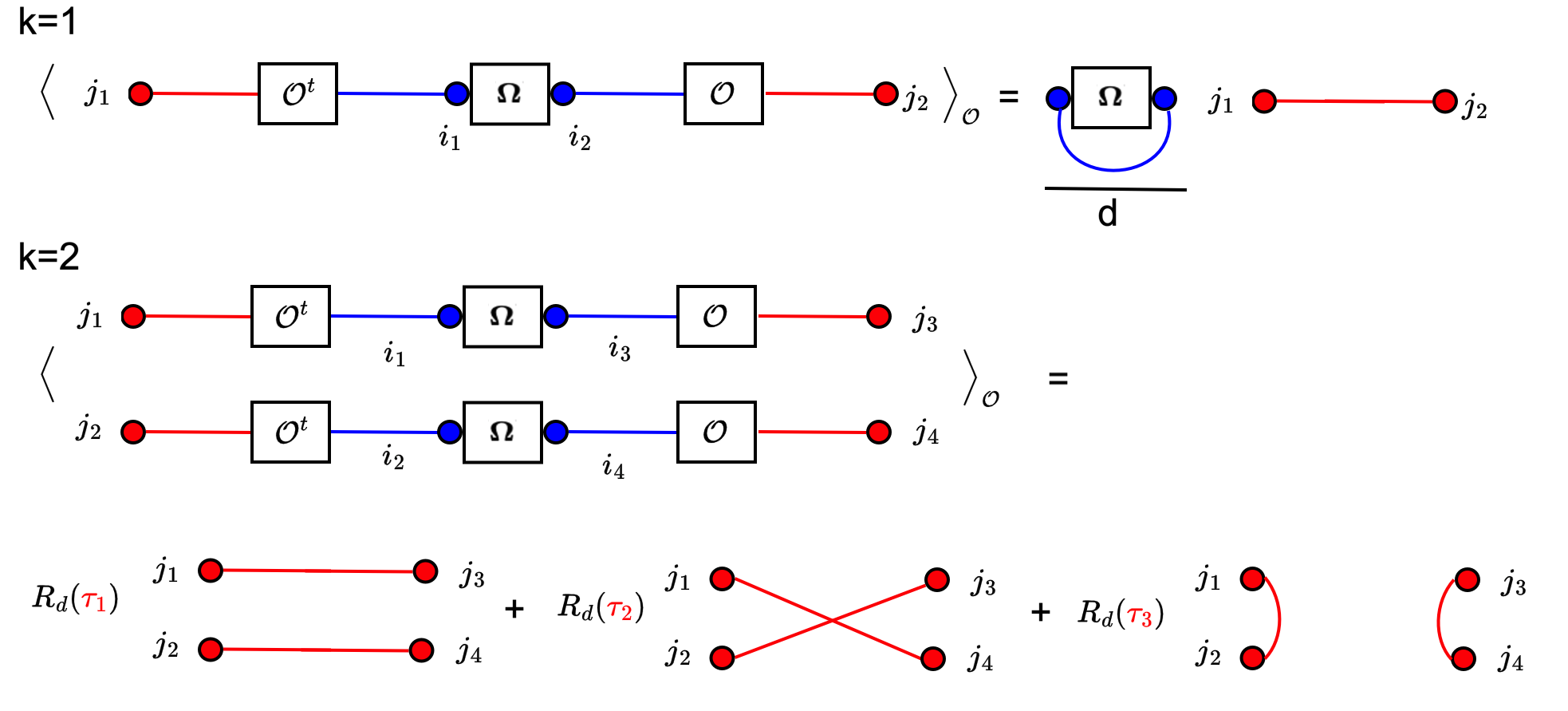}
    \caption{Graphical representation of the orthogonal twirling at the first and second order. At the first, the only contribution is proportional to the the identity operator.}
    \label{fig:IsoTwirlk1k2}
\end{figure}

\begin{figure}
    \centering
    \includegraphics[scale=0.25]{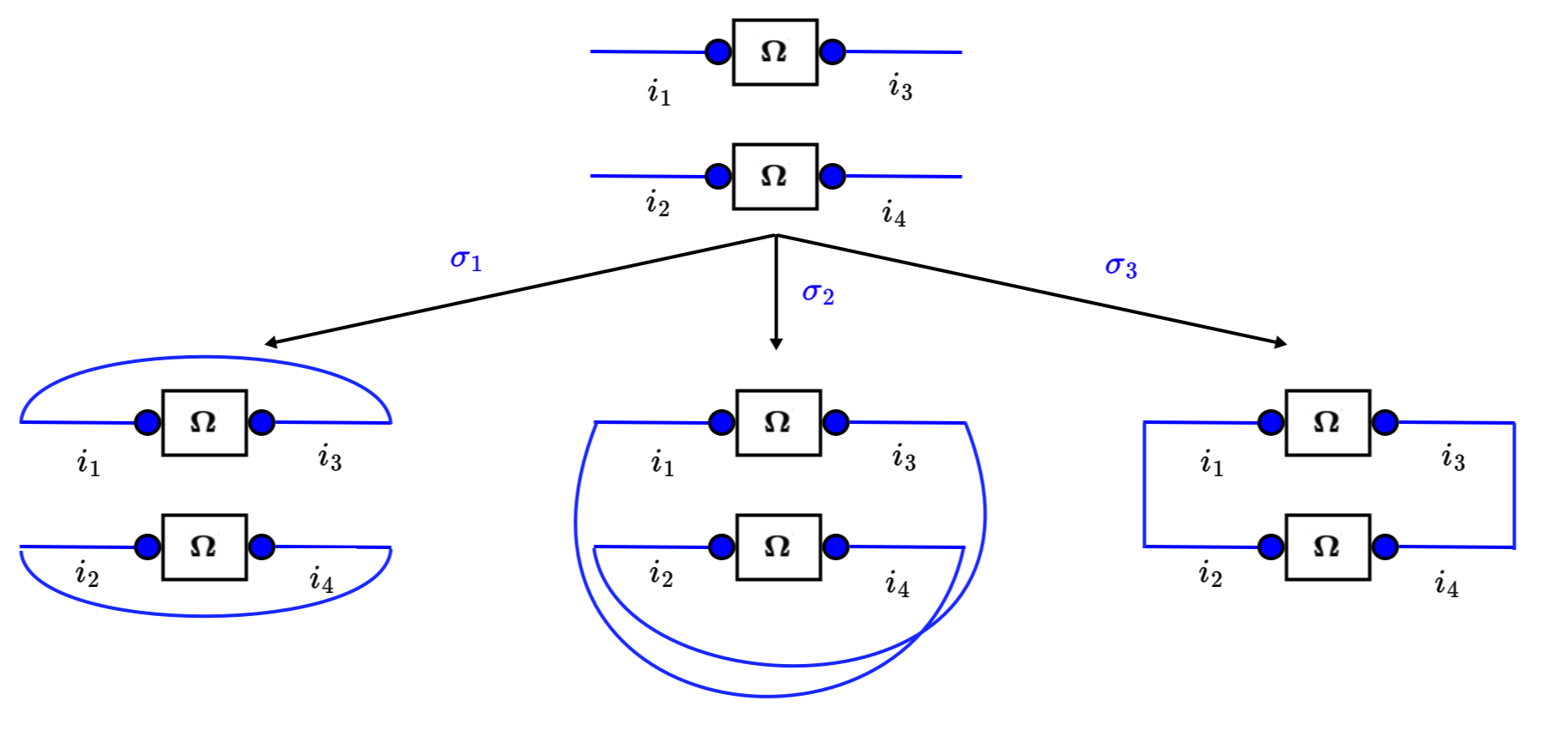}
    \caption{Contractions of the internal lines at the second order of the twirling. We see that given the pairing $\sigma_1$ contributes a term $\Tr(\mOmega)^2$, while $\sigma_2$ and $\sigma_3$ both contribute $\Tr(\mOmega^2)$.}
    \label{fig:IsoTwirlk2Om}
\end{figure}


\begin{figure}[hb!]
    \centering
    \includegraphics[scale=0.3]{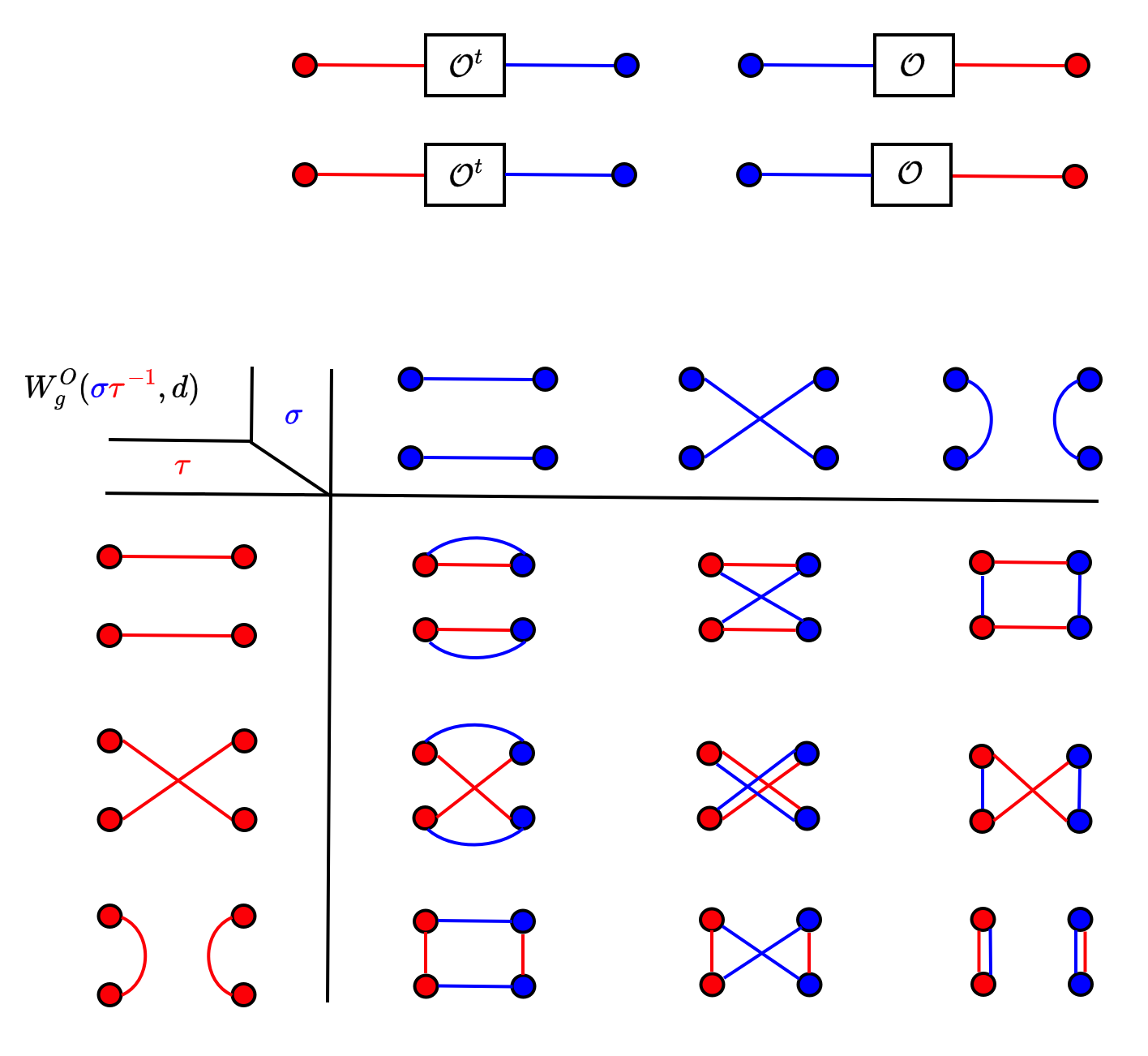}
    \caption{The Gram matrix $Gm$ from which the Weingarten coefficients are derived via a pseudo-inverse. The element $Gm(\sigma,\tau)$ is the number of loops produced by superposing $\sigma,\tau$ \cite{zinnjustin}. }
    \label{fig:wgo2}
\end{figure}

\subsection{Perturbative average}
Obtaining exact expressions up to an arbitrary order is rather cumbersome. For this reason, let us focus on  expressions up to order $k=2$.
In this case, we have

\begin{eqnarray}
    \langle \frac{d \vec g}{d\tau}\rangle_{\mO}&=&  \vec s \nonumber \\ 
    &+&\chi e^{\tau}  \langle\mOt \mathbf \mOmega \mO\rangle_{\mO} \mathbb G \vec s \nonumber \\
    &+&\chi^2e^{2\tau}\Tr_{2} \big( \langle(\mOt)^{\otimes 2} \mathbf \mOmega^{\otimes 2} \mO^{\otimes 2} \rangle_{\mO}\mathbb G^{\otimes 2}\mathbb S\big)\vec s \nonumber \\
    &&+O(\chi ^3)\nonumber 
    \label{eq:secondocc}
\end{eqnarray}

We see in the expression above that in order to perform this calculation, we need to calculate explicitly
\begin{eqnarray}
    \langle\mOt \mathbf \mOmega \mO\rangle_{\mO}&=& \mathbb M_1 \nonumber \\
    \langle(\mOt)^{\otimes 2} \mathbf \mOmega^{\otimes 2} \mO^{\otimes 2} \rangle_{\mO}&=& \mathbb M_2 
\end{eqnarray}

Using the methodology described in the previous section, it is easy to see that $\mathbb M_1=\frac{\Tr (\mOmega)|}{E}\mI$.

We consider the average up to the second order in $\chi$, 
eqn. (\ref{eq:secondocc}). Let us note that since $\mOmega^2=\mOmega$, we have $\Tr(\mOmega^2)=\Tr(\mOmega)$.

From eqn. (\ref{eqapp:secondorder}), we know instead that
\begin{eqnarray}
    \mathbb M_2= R_1 \mI^{\otimes 2}+ R_2 \mathbb S+R_3 \Pi.
\end{eqnarray}
The result of these contractions can be seen in a graphical representation in Fig. \ref{fig:secondordertwirl}.
Inserting these expressions into eqn.(\ref{eq:secondocc2}), we obtain
\begin{widetext}
\begin{eqnarray}
    \langle \frac{d \vec g}{d\tau}\rangle_{\mO}&=&  (\mI+\chi e^{\tau}  \frac{\Tr(\mOmega)}{E} \mathbb G)\vec s  +\chi^2e^{2\tau}\Tr_{2} \big((R_1 \mI^{\otimes 2}+ R_2 \mathbb S+R_3 \Pi) \mathbb G^{\otimes 2}\mathbb S_{12}\big)\vec s +O(\chi ^3)\nonumber \\&=&(\mI+\chi e^{\tau}  \frac{\Tr(\mOmega)}{E} \mathbb G)\vec s+\chi^2 e^{2 \tau}\Big((R_1+R_3) \mathbb G^2 + R_2 \Tr(\mathbb G)\mathbb G \Big)\vec s+O(\chi^3) \\
    &\approx &(\mI+\chi e^{\tau}  \frac{\Tr(\mOmega)}{E} \mathbb G)\vec s\nonumber \\
    &+&\chi^2 e^{2 \tau}\Big(((C_1+C_2) \Tr(\mOmega)^2+(3C_2+C_1)\Tr(\mOmega)) \mathbb G^2 + (C_2\Tr(\mOmega)^2+(C_1+C_2)\Tr(\mOmega)) \Tr(\mathbb G)\mathbb G \Big)\vec s \\
    \label{eq:secondocc2}
\end{eqnarray}
\end{widetext}

Using the assumption about the size of the system, e.g. that $E\gg 1$, and the fact that the graph is loop dense, e.g. $\Tr(\mOmega)=L \gg 1$, we have $C_1\approx \frac{1}{E^2}$ and $C_2\approx -\frac{1}{E^3}$, From which we get

\begin{eqnarray}
    \lim_{E\gg 1,L \gg 1}\langle \frac{d \vec g}{d\tau}\rangle_{\mO}
    &\approx &(\mI+\chi e^{\tau}  \frac{L}{E} \mathbb G)\vec s\nonumber \\
    &&+\chi^2 e^{2 \tau}\Big( (\frac{L^2}{E^2}+\frac{L}{E^2})\mathbb G^2  \nonumber \\
    &&\ \ \ + (- \frac{L^2}{E^3}+\frac{L}{E^2}) \Tr(\mathbb G)\mathbb G \Big)\vec s 
    \label{eqapp:secondocc}
\end{eqnarray}

We now define $\bar x(t)=\frac{1}{E} \sum_{j=1}^E x_i (t)$. It is easy to see then that $e^{\tau}\Tr(\mathbb G)/E=\bar x$. Summing the left and right equation, and performing the mean field replacement $\vec s\rightarrow \bar s \vec 1$, we then obtain the mean-field approximation, using the fact that $L/E^2\rightarrow 0$, we obtain
\begin{eqnarray}
    \lim_{E\gg 1,L \gg 1}\langle \frac{d \vec x}{d\tau}\rangle_{\mO}
    &\approx &(\mI+\chi   \frac{L}{E} \mX)\vec s\nonumber \\
    &+&\chi^2 \frac{L^2}{E^2}\Big( \mX^2 -  \frac{\Tr(\mX)}{E}\mX \Big)\vec s 
    \label{eqapp:thirdcc}
\end{eqnarray}
From summing over each index on the left, and dividing by $E$, we obtain
\begin{eqnarray}
    \frac{d \bar x}{dt}\approx \big(1+\chi \frac{L}{E} \bar x+ \chi^2 \frac{L^2}{E^2} (\overline{x^2} -{\bar x}^2)\big)\bar s
\end{eqnarray}
where $\overline{x^2}=\frac{1}{E}\sum_i x_i^2$. We introduce $\text{var}(x)=\overline {x^2}-{\bar x}^2$. Using this definition, we then have, for $E\gg 1$, that, introducing $\rho=\frac{L}{E}$, that

\begin{eqnarray}
    \frac{d \bar x}{dt}\approx \Big(1+\rho\chi  \bar x+ \chi^2 \rho^2 \text{var}(x) \Big)\bar s
\end{eqnarray}
which is the first equation used in the main text.
We see however that this equation is not closed, as it involves implicitly $\overline{x^2}$.
We then note that
\begin{eqnarray}
   \frac{d}{dt} (\vec x^2) =\frac{d}{dt} \mX \vec x&=&2 \mX \frac{d}{dt} \vec x\approx2(\mX+\chi   \frac{L}{E} \mX^2)\vec s\\
   &+&2\chi^2 \frac{L^2}{E^2}\Big( \mX^3 -  \frac{\Tr(\mX)}{E}\mX^2 \Big)\vec s 
\end{eqnarray}
From which we get, using the mean field approximation again for $\vec s $, and summing cleverly on the left, we get
\begin{eqnarray}
    \frac{d}{dt} \overline{x^2}\approx 2(\bar x+\chi \rho \overline{x^2})\bar s+2 \chi^2 \rho^2 (\overline{x^3}-\overline{x^2}\overline{x})  \bar s
\end{eqnarray}
which now depends on higher moments, implying a tower of coupled equations. We can impose closure by assuming that
\begin{eqnarray}
    \frac{d}{dt} \overline{x^k}=0.
\end{eqnarray}
 We will imposing closure at $k=3$, imposing $\overline{x^3}=r_3$ constant.
Note also that
\begin{eqnarray}
    \frac{d}{dt} \bar x^2&=&2\bar x\frac{d}{dt}\bar x\approx 2\bar x\Big(1+\rho\chi  \bar x+ \chi^2 \rho^2 \text{var}(x) \Big)\bar s \nonumber \\
    &=&2\Big(\bar x+\rho\chi  \bar x^2+ \bar x\chi^2 \rho^2 \text{var}(x) \Big)\bar s
\end{eqnarray}
and thus, using $\text{var}(x)=\overline{x^2}-\bar x^2$, we get
\begin{eqnarray}
    \frac{d}{dt} \text{var}(x)&=&2(\bar x+\chi \rho (\text{var}(x)+\bar x^2))\bar s\nonumber \\
    &+&2 \chi^2 \rho^2 (r_3-(\text{var}(x)+\bar x^2)\overline{x})  \bar s \nonumber \\
    &-&2\Big(\bar x+\rho\chi  \bar x^2+ \bar x\chi^2 \rho^2 \text{var}(x) \Big)\bar s\nonumber \\
    &=&2 \rho \chi \bar s\Big(\text{var}(x)(1-3 \bar x\rho \chi)+2 \rho \chi (r_3-\bar x^3)\Big)\nonumber \\
\end{eqnarray}

Thus, we obtain the set of coupled differential equations
\begin{eqnarray}
    \frac{d \bar x}{dt}&\approx& \Big(1+\rho\chi  \bar x+ \chi^2 \rho^2 v \Big)\bar s \label{eq:eqsbarxvarx1} \\
    \frac{dv}{dt} &\approx&2 \rho \chi \bar s\Big((1-3 \bar x\rho \chi)v+2 \rho \chi (r_3-\bar x^3)\Big)
    \label{eq:eqsbarxvarx2}
\end{eqnarray}

The question is whether now these equations have physical fixed points. However, it is not hard to see that there is no attractive fixed point in the dynamical system at finite $\bar x$ and $v$. This is shown in Fig. \ref{fig:varx}, where we plot the phase portrait of the dynamics. Note that the dynamics should be constrained in the box $\bar x\in [0,1]$. Moreover, 
from the Bhatia-Davis inequality \cite{bhatiadavis}, since $0\leq g_i\leq 1$, we have from the identity $\text{var}(g)=\frac{1}{E}(\Tr(\mathbb G^2)-\Tr(\mathbb G)^2)$
\begin{eqnarray}
    \text{var}(g)\leq (1-\bar g)\bar g\leq 1.
\end{eqnarray}
In particular
\begin{eqnarray}
    \frac{1}{E} \sum_i g_i^2 \leq \bar g,
\end{eqnarray}
or
\begin{eqnarray}
    \Tr(\mathbb G^2)\leq \Tr(\mathbb G).
\end{eqnarray}
However, the trajectories are not constrained to this box, and thus this must imposed in the variance. More importantly, we can see that the only fixed point is a saddle, which is unphysical as we know that there must be an attractive fixed point. Modifying the parameters of the equation only modifies the location of such a fixed point, but not the spectrum of the Jacobian.
We interpret this failure as the necessity to obtain non-perturbative results, which however we can only obtain using the asymptotics of the Weingarten calculus.

\begin{figure}
    \centering
    \includegraphics[scale=0.6]{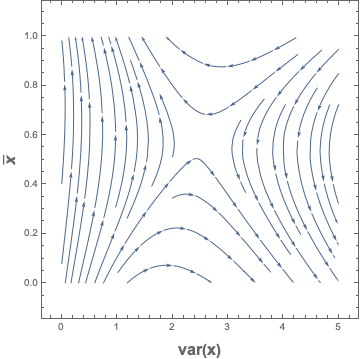}
    \caption{Phase portrait for the dynamical system of eqns. (\ref{eq:eqsbarxvarx1}-\ref{eq:eqsbarxvarx2}), for $r_3=0.4$, $\rho=\chi=0.8$, $s_0=0.1$. As we can see there is a saddle fixed point at an unphysical value of the variance. }
    \label{fig:varx}
\end{figure}

\subsection{Asymptotic regime}
As we have seen, the asymptotic results obtained in perturbation theory expanding in the power of $\chi$ lead to interesting but unsatisfactory results.
In order to obtain non-perturbative results, we are then forced to consider the asymptotic results in the Weingarten calculus. In particular, we will use the following result 
 \cite{COLLINSas}:


\begin{eqnarray}
    \lim_{E\rightarrow \infty} W_g^O(\textcolor{blue}{\sigma}\textcolor{red}{\tau^{-1}},E)=\frac{1}{E^k} \delta_{\textcolor{blue}{\sigma},\textcolor{red}{\tau}}+O\big(\frac{1}{E^{k+1}}\big).
\end{eqnarray}
This can be intuitively obtained from the definition of the Gram matrix. The dominant elements of the Gram matrix are in fact on the diagonal, as these always contribute $E^k$, while the off-diagonal elements are of the form $E,E^2,\cdots,E^{k-1}$.

Inserting this expression into eqn. (\ref{eqapp:isotwird}), we then obtain that \cite{banica}
\begin{eqnarray}
    R^{as}_{E} (\textcolor{red}{\tau},\mOmega)&=&\sum_{\textcolor{blue}{i_1,\cdots,i_{2k}}}\sum_{\textcolor{blue}{\sigma}\in \mathcal P_{2k}}\Delta_{\textcolor{blue}{\sigma}} (\textcolor{blue}{\vec i})\frac{1}{E^k} \delta_{\textcolor{blue}{\sigma},\textcolor{red}{\tau}}\mOmega_{\textcolor{blue}{i_1}\textcolor{blue}{i_{k+1}}}\cdots \mOmega_{\textcolor{blue}{i_k}\textcolor{blue}{i_{2k}}} \nonumber\\
    &=&\sum_{\textcolor{blue}{i_1,\cdots,i_{2k}}}\Delta_{\textcolor{red}{\tau}} (\textcolor{blue}{\vec i})\frac{1}{E^k} \mOmega_{\textcolor{blue}{i_1}\textcolor{blue}{i_{k+1}}}\cdots \mOmega_{\textcolor{blue}{i_k}\textcolor{blue}{i_{2k}}}
    \label{eqapp:isotwirdas}
\end{eqnarray}
\begin{eqnarray}
    &&(\hat{\mathcal{R}}_{\mathbb O}^{(2k)})_{\textcolor{red}{j_1\cdots j_k}}^{\textcolor{red}{j_{k+1}\cdots j_{2k}}}(\mOmega)\nonumber \\
     &&\hspace{1cm}= \sum_{\textcolor{red}{\tau}\in \mathcal P_{2k}} \Delta_{\textcolor{red}{\tau}} (\textcolor{red}{\vec j}) R^{as}_E (\textcolor{red}{\tau},\mOmega)\nonumber \\
     &&\hspace{1cm}=\frac{1}{E^k}\sum_{\textcolor{red}{\tau}\in \mathcal P_{2k}} \Delta_{\textcolor{red}{\tau}} (\textcolor{red}{\vec j}) \sum_{\textcolor{blue}{i_1,\cdots,i_{2k}}}\Delta_{\textcolor{red}{\tau}} (\textcolor{blue}{\vec i}) \mOmega_{\textcolor{blue}{i_1}\textcolor{blue}{i_{k+1}}}\cdots \mOmega_{\textcolor{blue}{i_k}\textcolor{blue}{i_{2k}}}.\nonumber \\
     \label{eqapp:asympre}
\end{eqnarray}

Let us now look at the expressions
\begin{eqnarray}
    &&\Tr_{2\cdots k}\Big(\langle \mO^{t,\otimes k} \mOmega^{\otimes k} \mO^{\otimes k}\rangle_\mO\mathbb G^{\otimes k}\mathbb S_{2\cdots k}\Big)\nonumber \\
    &&\hspace{2cm}=\Tr_{2\cdots k}\Big(\mathbb M_k\mathbb G^{\otimes k}\mathbb S_{2\cdots k}\Big)
\end{eqnarray}
Let us now call $\mathbb B_\tau$ the operator corresponding to the element $\tau$ of the Brauer algebra associated to the contractions $\Delta_{\textcolor{red}{\tau}}(\textcolor{red}{\vec j})$. 
Using eqn. (\ref{eqapp:asympre}), we obtain the asymptotic expression
\begin{eqnarray}
     &&\Tr_{2\cdots k}\Big(\mathbb M_k\mathbb G^{\otimes k}\mathbb S_{2\cdots k}\Big)\nonumber \\
     &&\hspace{1cm}=\frac{1}{E^k}\sum_{\textcolor{red}{\tau}\in \mathcal P_{2k}} \Tr_{2\cdots k}\Big(\mathbb B_{\textcolor{red}{\tau}}\mathbb G^{\otimes k}\mathbb S_{2\cdots k}\Big)\Tr\Big(B_{\textcolor{red}{\tau}}\mOmega^{\otimes k} \Big)\nonumber 
\end{eqnarray}
Now, note that $\mathbb S_{2\cdots k}$ is associated to an element of the Brauer algebra, with pairing $\{\{1,2k\},\{2,k+1\},\{3,k+2\},\cdots,\{k,2k-1\}\}$. Let us call this element $\textcolor{red}{\rho}$. Then, we can write the expression above as
\begin{eqnarray}
     &&\Tr_{2\cdots k}\Big(\mathbb M_k\mathbb G^{\otimes k}\mathbb S_{2\cdots k}\Big)\nonumber \\
     &&\hspace{0.5cm}=\frac{1}{E^k}\sum_{\textcolor{red}{\tau}\in \mathcal P_{2k}} \Tr_{2\cdots k}\Big(\mathbb B_{\textcolor{red}{\rho\tau}}\mathbb G^{\otimes k}\Big)\Tr\Big(B_{\textcolor{red}{\tau}}\mOmega^{\otimes k} \Big)
\end{eqnarray}
Now, note that we can write
\begin{eqnarray}
    \Tr\Big(B_{\textcolor{red}{\tau}}\mOmega^{\otimes k} \Big)=\prod_{s=1}^{c(\textcolor{red}{\tau})}\Tr(\mOmega^{j_{s}(\textcolor{red}{\tau})}), \ \ \ \sum_{s=1}^{c(\tau)}{j_s(\textcolor{red}{\tau})}=k,
\end{eqnarray}
where $c({\textcolor{red}{\tau}})$ is the number of connected components of $\tau$.
Since we have $\mOmega^k=\mOmega$, the expression above is maximized when $c(\textcolor{red}{\tau})=k$. This is true when $\textcolor{red}{\tau}$ is the identity over the Brauer algebra. As a result, we have proven the following 
\begin{proposition} \label{lemma:l3}
     If $\Tr(\mOmega)=L>1$, then $\text{max}_{\tau} \Tr\Big(B_{\tau}\mOmega^{\otimes k} \Big)=\Tr\Big(\mOmega \Big)^k$, obtained for $\tau=e$, the identitity in the Brauer algebra.
\end{proposition}

We now want to argue that if $L(E)$ grows as a function of $E$, then we have
\begin{eqnarray}
     \lim_{E\rightarrow \infty} \Tr_{2\cdots k}\Big(\mathbb M_k\mathbb G^{\otimes k}\mathbb S_{2\cdots k}\Big)&=&\frac{L^k}{E^k} \Tr_{2\cdots k}\Big(\mathbb B_{\textcolor{red}{\rho}}\mathbb G^{\otimes k}\Big) \nonumber \\
     &=& \frac{L^k}{E^k}\mathbb G^k+O(\frac{L^{k-1}}{E^k}).
\end{eqnarray}
which is the expression we need in order to derive the large $E$ mean field theory.
\begin{figure}
    \centering
    \includegraphics[scale=0.3]{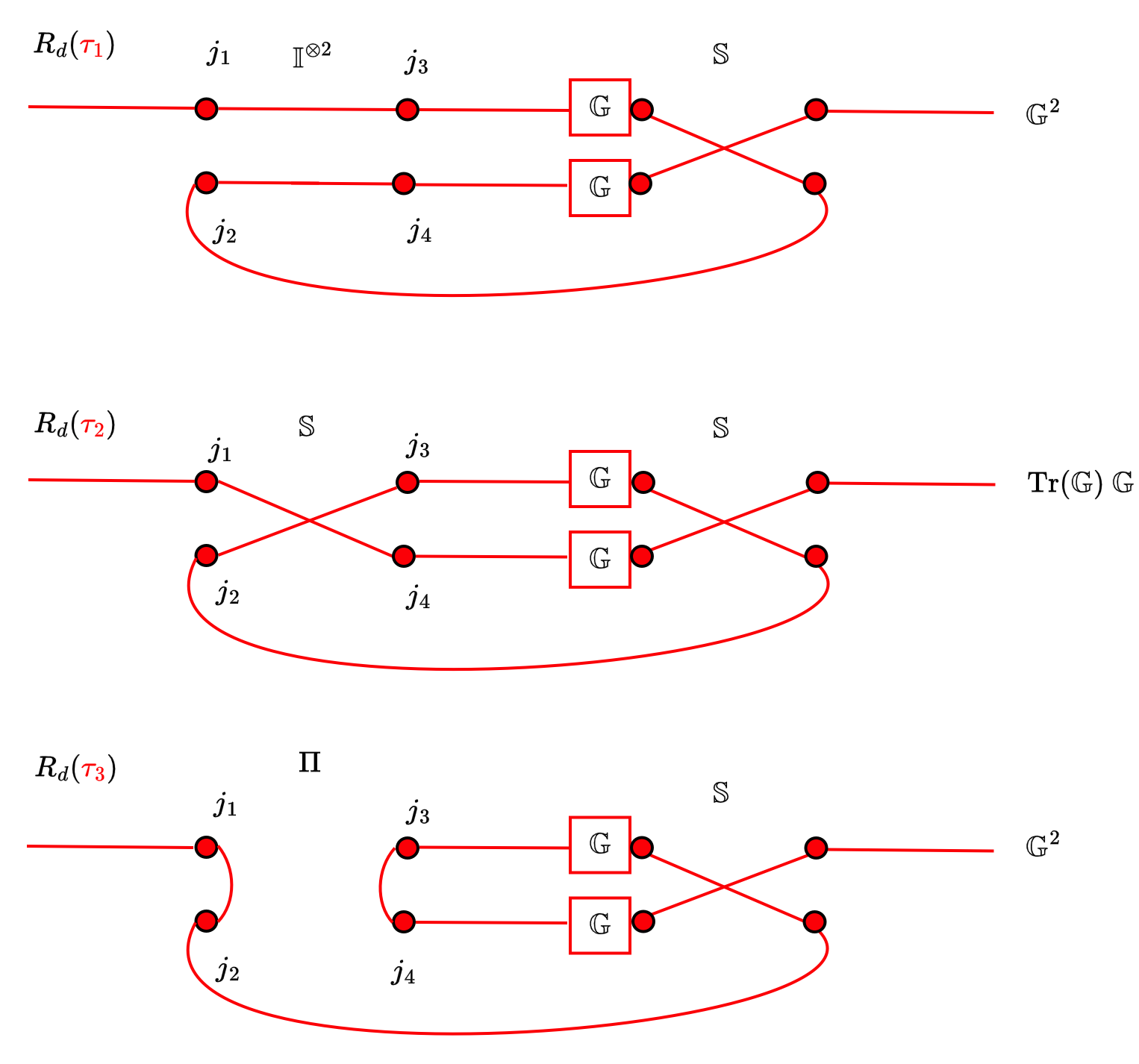}
    \caption{Tensorial representation of the partial trace occurring after the twirl orthogonal average at the order $\chi^2$. As we can see, the first and third terms give $\mathbb G^2$, while the middle term contributes $\Tr(\mathbb G)\mathbb G$.}
    \label{fig:secondordertwirl}
\end{figure}

We now want to show why this is the dominant contribution for each term at the order $\chi^k$. Let us now discuss quantities of the form $\Tr(\mathbb G^j)$.
First, note that we have that $|x_i|\leq 1$, and thus, since $\alpha\geq 0$, we have $|g_i|=|e^{-\alpha t}x_i|\leq 1$ for $t\geq 0$. Note that $\Tr(\mathbb G)= E \bar g$. It then means that 
\begin{eqnarray}
    0\leq \Tr(\mathbb G)\leq E,
\end{eqnarray}
and thus $0\leq \bar g\leq 1.$
Note that in general, $\Tr(\mathbb G^j)=\sum_i g_i^{j}$.

\begin{proposition}
   Let $\chi<1$, $0\leq |\mathbb X_{ii}|\leq 1$. Then
\begin{eqnarray}
    \lim_{E \rightarrow \infty}\langle (I-\chi \mO^t \mOmega \mO \mX)^{-1}\rangle_{\mO} = (I-\chi \langle\mO^t \mOmega \mO\rangle_{\mO} \mX)^{-1}. \nonumber
\end{eqnarray}
\end{proposition}

\begin{proof}
Since both $x_i$ and $g_i$ are constrained to the interval $[0,1]$, such result holds for both $\mathbb G$ and $\mathbb X$. Note that we have, since $g_i\in [0,1]$, $g_i^r\leq g_i$, $\forall r\geq 1$. It follows that \cite{polya,Bernstein}, 
\begin{eqnarray}
    \frac{\Tr(\mathbb G^r)}{E^r}\leq \frac{\Tr(\mathbb G)}{E^r} \leq \frac{\bar g}{E^{r-1}}\leq \frac{1}{E^{r-1}},\label{eqapp:ineq}
\end{eqnarray}
since $\Tr(\mathbb G)\leq E$.

Let us now analyze terms of the form
\begin{eqnarray}
    \Tr_{2\cdots k}\Big(\mathbb B_{\textcolor{red}{\rho\tau}}\mathbb G^{\otimes k}\Big)
\end{eqnarray}
for $\tau\neq e$. This trace can be written in the form
\begin{eqnarray}
    \prod_{j=1}^s\Tr(\mathbb G^{m_{j}})\mathbb G^{m_0}
\end{eqnarray}
where $m_0+\sum_{j=1}^s m_j =k$, where $s$ is the number of closed loops in the Brauer diagram resulting in the partial trace, and $m_j$ is the number of $\mathbb G$ inserted in the loop $j$.
Using eqn. (\ref{eqapp:ineq}), we have
 \begin{eqnarray}
    \prod_{j=1}^s\frac{\Tr(\mathbb G^{m_{j}})}{E^{m_j}}&\leq&  \prod_{j=1}^s\frac{\Tr(\mathbb G)}{E^{m_j}}=\frac{\bar g^s}{E^{\sum_{j=1}^s m_j-s}}=\frac{\bar g^s}{E^{k-m_0-s}}\nonumber 
\end{eqnarray}
Now note that since $m_j\geq 1$, we have $k-m_0-s\geq 0$.  Note that we assume $\bar g< 1$, which is the interesting case, as $\bar g=1$ means that all memristors are at the fixed point, and the system is not evolving.
As a result, we have shown that every loop contribution, for $\bar g<1$ will be sub-dominant with respect to the identity element in the Brauer algebra. Since identity is the only term not containing loops, the result follows in the limit $E\rightarrow \infty$. It then follows that
\begin{eqnarray}
    \lim_{E\rightarrow \infty} \Tr_{2\cdots k}\Big(\mathbb M_k\mathbb G^{\otimes k}\mathbb S_{2\cdots k}\Big)=\frac{L^k}{E^k} \mathbb G^k+O(\frac{c_k}{E^{k+1}}).
\end{eqnarray}
where $c_k$ is a constant.

Now, the question is whether we can swap the series with such limit. Note that the (scalar) series
\begin{eqnarray}
    \frac{1}{1-x}=\sum_{j=1}^\infty x^j
\end{eqnarray}
converges uniformly on $x\in[0,1)$. Similarly, the von Neumann series $(I-A)^{-1}$ converges uniformly if $\forall \lambda\in \Lambda(A)$, $\lambda<1$. Note then that since  $|x_i|\leq 1$ and $\chi<1$, we have $\Lambda(\chi \mO^t  \mOmega\mO \mX)\in R(0,\chi)$, where $R(0,a)$ is the area in the complex plane of radius $a$ and centered in zero.
Using Proposition \ref{lemma:l3}, the contribution due to the $\mOmega$ traces is obtained. Then, at order $k$, the dominant operator contribution is of the form $\mathbb G^k \Tr(\mOmega)^k$. Going to the $\mathbb X$ variables,
we have then obtained the final result
\begin{eqnarray}
    \lim_{E \rightarrow \infty}\langle (I-\chi \mO^t \mOmega \mO \mX)^{-1}\rangle_{\mO} &=&\sum_{k=0}^\infty \lim_{E \rightarrow \infty} \langle (\chi \mO^t  \mOmega\mO \mX)^k\rangle_{\mO}\nonumber \\
    &=& (I-\chi \langle \mO^t \mOmega \mO\rangle_{\mO} \mX)^{-1}\rangle_{\mO}\nonumber
\end{eqnarray}
which proves the proposition.
\end{proof}

\begin{figure}
    \centering
    \includegraphics[scale=0.34]{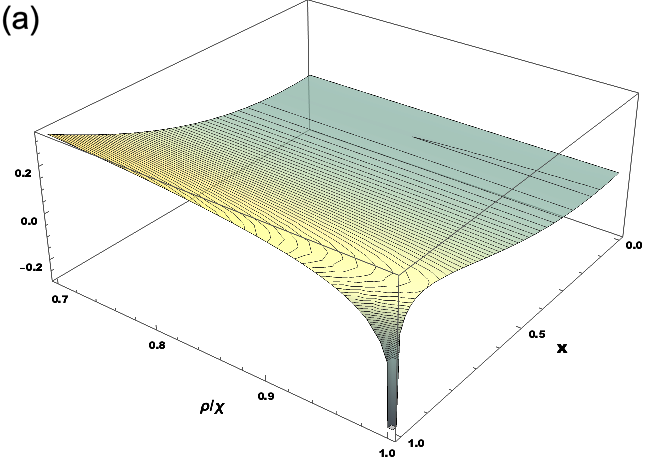}\\
    \includegraphics[scale=0.34]{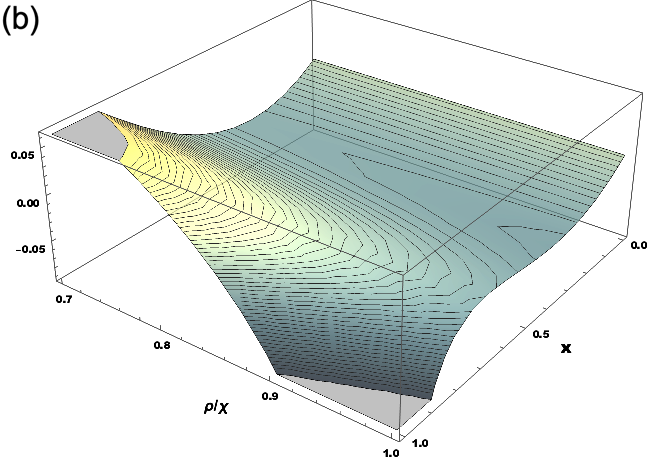}\\\includegraphics[scale=0.34]{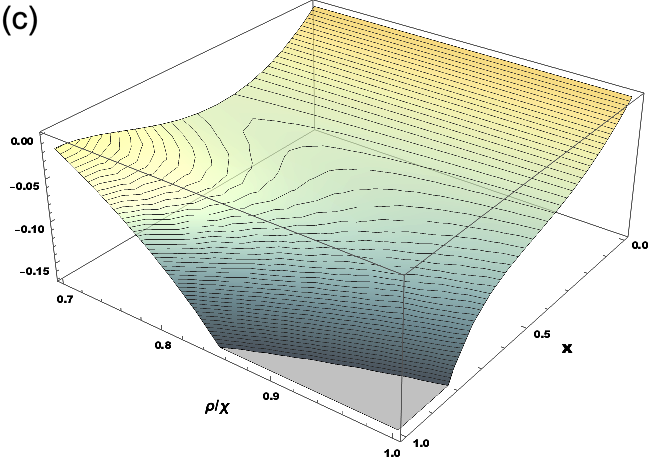}
    \caption{Effective mean field potential as a function of $x$ and $\rho/\chi$ for the values of (a) $s=0.1$ (b) $s=0.23$ and (c) $s=0.3$. As we can see, for values of $\rho/\chi$ above $\approx 0.8$, the potential develops two competing minima as we increase $s$, the effective voltage. For lower values, instead, the single minimum moves smoothly as a function of $s$. }
    \label{fig:effpot}
\end{figure}

We use the proposition above to obtain  the average dynamics over each cycle class $C(E,L)$. We have
\begin{eqnarray}
   \lim_{E\rightarrow \infty} \big \langle \frac{d \vec g}{d\tau}\big\rangle&\approx & \sum_{k=0}^\infty  \chi^ke^{k \tau} \frac{\Tr(\mathbf \mOmega)^k}{E^k} \Tr_{2\cdots k} \big( \mathbb G^{\otimes k}\mathbb S_{2\cdots k}\big) \vec s \nonumber \\
   &=&\sum_{k=0}^\infty  \chi^ke ^{k \tau } \frac{\Tr(\mathbf \mOmega)^k}{E^k} \mathbb G^k\vec s \nonumber \\
   &=&\frac{1}{\mI-e^{\tau} \frac{\Tr(\mathbf \mOmega) \chi}{E} \mathbb G} \vec s
\end{eqnarray}

A result which refines those of \cite{caravelliscience,BCCV} although for particular type of random matrices and only in the large $E$ limit.

Because of this, this implies that each memristive device is decoupled from another one on average, as $\mathbb G$ is a diagonal matrix. We can then analyze the behavior of the system by simply looking independently at each equation, and we get for the $\langle x_i\rangle$ variables that

\begin{eqnarray}
    \frac{d\langle  x_i\rangle}{d\tau}  &=&- \langle x_i\rangle + \frac{s_i}{\alpha \beta} \frac{ 1}{1- \chi \frac{L}{E} \langle x_i\rangle}\\
    &=&-\partial_{\langle x_i\rangle} V(\langle x_i\rangle)
    \label{eq:mft}
\end{eqnarray}
where $L=\Tr(\mathbf \mOmega)<E$ is the number of cycles in the graph. We have that $L=E$ if the number memristors is equal to the number of cycles, and thus if all memristors are completely decoupled. In this case, $\mathbf \mOmega=I$ and in fact equation (\ref{eq:mft}) is exact.
Note at this point that we have the effective mean field potential

\begin{eqnarray}
    V(\langle x_i\rangle)=\frac{1 }{2}\langle x_i\rangle^2+\frac{s_i  }{ 
 \alpha\beta} \frac{E}{\chi L}\log (1-\frac{L \chi}{E}  \langle x_i\rangle) 
 \label{eq:mft2}
\end{eqnarray}

We see the resemblance between eqn. (\ref{eq:mft2}) and eqn. (\ref{eq:meanfielddynp}). The equation is identical as long as we replace $s_i\rightarrow \bar s$ and $\chi \rightarrow \chi \frac{L}{E}$, and thus we can simply analyze the behavior of the effective potential as a function of voltage as we did in previous papers.

Let us now make some comments on the effective value of $\chi$, using eqn. (\ref{eq:conntodegree}) introduced earlier. In the limit $E$ for a connected graph, we obtain that 
\begin{eqnarray}
    \chi_{eff}=\chi (1-\frac{2}{\bar d})
\end{eqnarray}
where $\bar d$ is the average degree.
This result allows us to connect the geometrical properties of the graph to the transition properties. For a connected graph, $\bar d \geq 2$. This implies that $1-\frac{2}{\bar d}\leq 1$. As a result, we have that if $\bar d\propto V$ as in the case of a complete graph, then in the limit $E\rightarrow \infty$ we have $\chi_{eff}\rightarrow \chi$. For planar graphs instead, we have $\bar d\leq 6$ because of Kuratowski's lemma. It follows that 
\begin{eqnarray}
    0 \leq  \frac{\chi{eff}}{\chi}\leq \frac{2}{3}
\end{eqnarray}

The effective mean field potential as a function of $\rho/\chi$ and three values of $s$ are shown in Fig. \ref{fig:effpot}. As we can see, for low value of $s$ the minimum is located at $x=0$ for all values of $\chi$. However, for larger values of $s$ the potential has a minimum at intermediate values of $x$. At large values of $s$, the potential has two competing minima only for values of $\rho/\chi$ above $0.8$. However, for quasi-planar graphs we know that these values are not attainable. This provides an explanation for the absence of first order \text{bulk} induced \cite{caravelliscience} transitions in the effective 2-probe conductance of silver nanowire experiments \cite{AtomicSwitch1,caravelliadp,hochstetteretal}.

\section{Conclusions}

In conclusion, this manuscript utilized techniques from the orthogonal Weingarten calculus to derive the mean field theory of memristive systems, with a specific application in mind of nanowires of low dimensionality. The results obtained shed light on the behavior of quasi two-dimensional systems, revealing that they exhibit only a crossover and not a first-order switching transition in conductance as a bulk phenomenon. This finding provides important insights into the fundamental physics governing memristive behavior in nanowires, and suggests that first order transitions are possibly a \textit{boundary} phenomenon, e.g. driven by avalanches in switching near the boundary for each memristive device. Of course, although our results apply to the case of a simple toy model, the techniques developed here can be used in more realistic case.

Furthermore, the developed technique has broader applicability in the theory of neuromorphic devices, as it draws connections with the physics of brain-like materials. This implies that the derived mean field theory can be extended to understand and potentially design other neuromorphic devices beyond nanowires, opening up new possibilities for the development of advanced electronic devices with brain-inspired functionalities.

The findings presented in this paper contribute to the understanding of memristive systems and their behavior in quasi two-dimensional systems, while also highlighting the broader applicability of the developed mean field theory to the field of neuromorphic devices. These results have the potential to impact the design and development of future electronic devices with applications in areas such as artificial intelligence, cognitive computing, and brain-computer interfaces. Further research and experimental validation of the proposed theory in various material systems are warranted to fully comprehend the potential of this approach in advancing the field of memristive devices and their applications. 
Additionally, the connection between memory and ergodicity in brain-inspired devices is worth noting in the context of the derived mean field theory of memristive systems using cycle classes. The understanding of how memory is encoded and processed in neuromorphic devices is crucial for the development of advanced brain-inspired computing systems \cite{hochstetteretal,memorynanw}. The insights obtained from the developed mean field theory provide valuable information on the role of boundary-induced transitions such as critical avalanches and ergodicity.

The findings suggest that the mean field theory can offer a deeper understanding of the interplay between memory, ergodicity, and the conductance switching behavior in quasi two-dimensional memristive systems  such as nanowire connectomes. This knowledge can be leveraged to design more efficient and reliable neuromorphic devices that mimic the memory processing mechanisms of brain-like devices.

\begin{acknowledgments}
The work of FC was carried out under the auspices of the NNSA of the U.S. DoE at LANL under Contract No. DE-AC52-06NA25396, and in particular support from LDRD via 20230338ER and 20230627ER. We thank Marco Cerezo, Salvatore Olivero,  Lorenzo Leone and Alioscia Hamma for many useful discussions on the Weingarten calculus.
\end{acknowledgments}

\bibliography{bibtex.bib}
\clearpage
\appendix

\section{Brief introduction to tensor products}\label{app:tensprod}
A matrix, or linear operator $M_{ij}$ will be denoted in the following as $M_j^i$, to make clear which indices are input and which indices are output. Matrix multiplication is then denoted by $\sum_j M_{ij} N_{jk}\equiv \sum_j M^i_j N^j_k$, indicating that contractions occur between bottom and upper indices respectively.
Matrix products between matrices, in general, can be written as
\begin{equation}
    (A_1 A_2)_{ij}=\sum_k (A_1)^i_{k} (A_2)^k_{j}=\sum_{kt} (A_1)^i_{k} (A_2)^k_{j}  \delta_{kt}
\end{equation}
The reason for the introduction of this notation will be clear in a moment, when we introduce tensor products.

We can write the identity above in a different way. We introduce the tensor product of the operator $A_1$ and $A_2$, and define
\begin{equation}
    A_1\otimes A_2:\mathbb{R}^E\otimes \mathbb{R}^E \rightarrow \mathbb{R}^E\otimes \mathbb{R}^E.
\end{equation}
Thus, instead of acting on a single copy of $\mathbb{R}^E$ it acts on two copies, $\mathbb{R}^E\otimes \mathbb{R}^E\equiv \mathbb{R}^{d,\otimes 2}$  (we will assume $d$ is the dimensionality of the single linear space from now on) e.g.

\begin{equation}
    (A_1\otimes A_2)(\vec v_1\otimes \vec v_2)=(A_1 \vec v_1)\otimes (A_2 \vec v_2).
\end{equation}
Any linear operator on $\mathcal R^{d,\otimes n}$ can be written as
\begin{eqnarray}
    A^{i_1\cdots i_n}_{j_1 \cdots j_n}&=&\sum_{j_1\cdots j_n,k_1 \cdots k_n=1}^n c^{j_1\cdots j_n}_{k_1\cdots k_n} A_{j_1k_1}\otimes \cdots \otimes A_{j_1 k_n}\nonumber \\
    &=&\sum_{j_1\cdots j_n=1}^{n^2} \tilde c_{j_1\cdots j_n} \tilde A_{j_1}\otimes \cdots \otimes \tilde A_{j_n}\\
\end{eqnarray}
where $\tilde A_{j_i}$ are a basis for $GL(E)$. We used the notation in which the lower index is contracted with the vector's index. Thus
\begin{eqnarray}
    A(\vec v_1\otimes \cdots \otimes \vec v_n)=\sum_{j_1,\cdots,j_n}A^{i_1\cdots i_n}_{j_1 \cdots j_n} (v_1^{j_1}\otimes \cdots \otimes v_n^{j_n})\nonumber 
\end{eqnarray}
As made explicit above, we use the notation for a tensor the bottom indices are input and the top indices are output.

We now introduce the swap operator on $\mathbb{R}^E\otimes \mathbb{R}^E$, $\mathbb{S}$, or $\mathbb{S}_{12}$ to introduce a notation that will be clear in a moment.
The operator $\mathbb S_{12}$ performs the following action

\begin{equation}
    \mathbb S (\vec v_1 \otimes \vec v_2)=\vec v_2 \otimes \vec v_1
\end{equation}

These operations can be written in graphical terms. Tensor products can be written in terms of lines and boxes, as in Fig. \ref{fig:identities}. The last operation we wish to discuss is the partial trace. Consider a tensor $A^{i_1\cdots i_n}_{j_1 \cdots j_n}$ acting on $\mathbb R^{\otimes n}$. A partial trace $\Tr_k(\cdot): GL(E)^{\otimes n}\rightarrow GL(E)^{\otimes n-1}$ is the operation $\Tr_k(A^{i_1\cdots i_n}_{j_1 \cdots j_n})=\sum_n A^{i_1\cdots i_{k-1} n i_{k+1} i_n}_{j_1 \cdots j_{k-1} n j_{k+1} \cdots j_n}$. This is a generalization of the trace of a matrix, in which $\Tr(A_{ij})=\Tr(A^i_{j})=\sum_n A_{nn}$. The partial trace can be generalized to multiple index contractions, depending on the subspace.

These operations are shown in Fig. \ref{fig:identities} in a graphical manner. Tensor products can be represented as lines and boxes, while swap operators simply line inversions. The partial or full trace contracts lines from the left to the equivalent line on the diagram. Contracted lines simply mean summing over the indices of that particular line. For instance, $\Tr(A)$ is shown in Fig. \ref{fig:TrA}.

In the main text, we use the combined swap operator
\begin{eqnarray}
    S_{2\cdots k}=S_{1,2}S_{2,3} S_{3,4}\cdots S_{k-1,k}
\end{eqnarray}
to construct
\begin{eqnarray}
    A^k=\text{Tr}_{2\cdots k}(A^{\otimes k} S_{2\cdots k}).
\end{eqnarray}
Such identity, for $k=4$, is shown graphically in Fig. \ref{fig:at4}.

\begin{figure}[b!]
    \centering
    \includegraphics[scale=0.3]{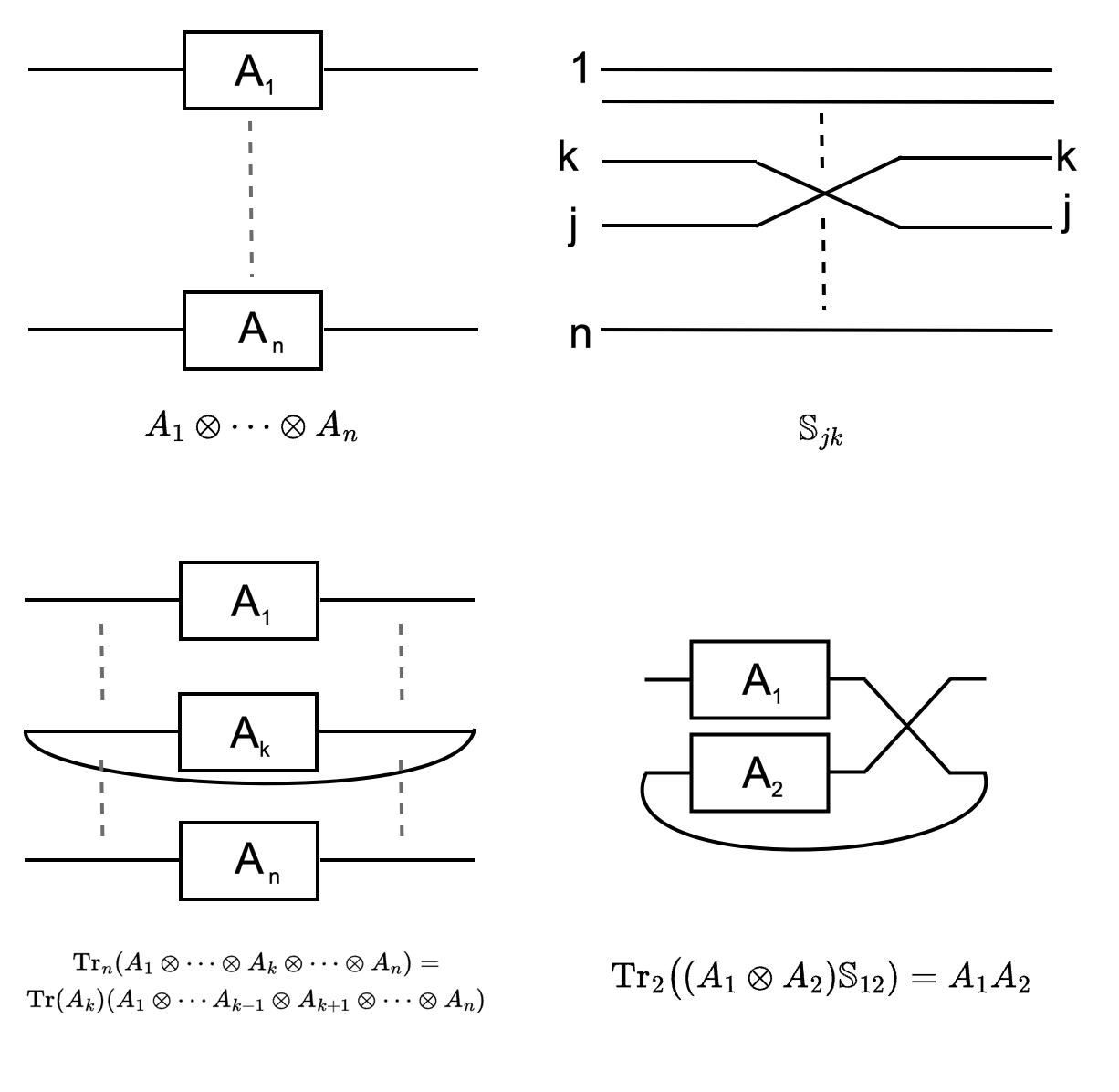}
    \caption{Graphical representations of the identities involving partial traces over tensor products of linear operators, and the swap operator.}
    \label{fig:identities}
\end{figure}

\begin{figure}[b!]
    \centering
    \includegraphics[scale=0.3]{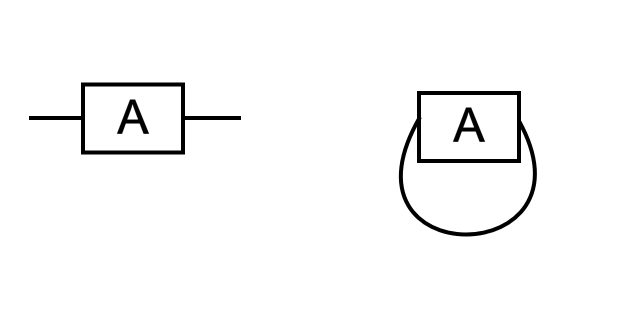}
    \caption{A linear operator $A$ and the graphical representation of its trace.}
    \label{fig:TrA}
\end{figure}


\end{document}